\documentclass[11pt,letterpaper]{article}

\usepackage{amsmath,amsfonts,graphpap,amscd,mathrsfs,graphicx,lscape}
\usepackage{amssymb,amstext,xspace}
\usepackage{algorithm,algorithmic}
\usepackage{paralist}

\usepackage[suppress]{color-edits}
\addauthor{vs}{green}
\addauthor{bl}{blue}
\addauthor{mf}{red}
\addauthor{ni}{magenta}
\addauthor{ri}{cyan}

\usepackage{amsthm}

\newcommand{\E}{\mathbb{E}}

\newcommand {\NP}{\mathcal {NP}}
\newcommand{\MOPH}{\mathcal{MPH}}
\newcommand{\PLE}{\mathcal{PLE}}
\newcommand{\XOS}{\mathcal{XOS}}
\newcommand{\nonnegR}{\mathbb{R}^+}
\newcommand{\SMD}{\mathcal{SMD}}
\newcommand{\PH}{\mathcal{PH}}

\newcommand{\items}{\mathop{\mathrm{items}}\nolimits}

\newtheorem{theorem}{Theorem}[section]
\newtheorem{lemma}[theorem]{Lemma}

\newtheorem{prop}[theorem]{Proposition}

\newtheorem{corollary}[theorem]{Corollary}
\newtheorem{observation}[theorem]{Observation}
\newtheorem{conj}[theorem]{Conjecture}

\newtheorem{defn}{Definition}[section]
\newtheorem{LP}[defn]{Linear Program}

\def\squareforqed{\hbox{\rule{2.5mm}{2.5mm}}}

\def\QED{\ifmmode\squareforqed 
  \else{\nobreak\hfil   
    \penalty50                 
    \hskip1em                  
    \null                      
    \nobreak                   
    \hfil                      
    \squareforqed              
    \parfillskip=0pt           
    \finalhyphendemerits=0     
    \endgraf}                  
  \fi}

\def\blksquare{\rule{2mm}{2mm}}
\def\qedsymbol{\blksquare}
\newcommand{\bg}[1]{\medskip\noindent{\bf #1}}
\newcommand{\ed}{{\hfill\qedsymbol}\medskip}

\newenvironment{proofof}[1]{{\it{Proof of #1 : }}}{\ed}

\newcommand{\R}{\ensuremath{\mathbb R}}

\newcommand{\A}{\ensuremath{\mathcal{A}}}

\newcommand{\T}{\ensuremath{\mathcal T}}

\newcommand{\Ell}{\ensuremath{\mathcal{L}}}

\newcommand{\comment}[1]{}
 {}

\newcommand{\junk}[1]{}

\newlength{\tmp} \newlength{\lpsx} \newlength{\lpsy} \newlength{\upsx} \newlength{\upsy}

\newcommand{\opt}{\text{\textsc{Opt}} }

\newcommand{\X}{\ensuremath{\mathcal{X}}}

\newcommand{\opti}{\ensuremath{X_{i}^*}}

\newcommand{\ESk}[1]{\ensuremath{#1 | _k}}

\newcommand{\rbid}{\ensuremath{{\bf b}}}
\newcommand{\rprice}{\ensuremath{{\bf p}}}

\newcommand{\vals}{\ensuremath{{\bf v}}}
\newcommand{\val}{\ensuremath{v}}

\newcommand{\Omit}[1]{}

\begin{document}

\title{A Unifying Hierarchy of Valuations\\with Complements and Substitutes}

\author{
Uriel Feige\thanks{Weizmann Institute of Science; {\tt uriel.feige@weizmann.ac.il}}
\and
Michal Feldman\thanks{Tel-Aviv University; {\tt mfeldman@tau.ac.il}}
\and
Nicole Immorlica\thanks{Microsoft Research; {\tt nicimm@gmail.com}}
\and
Rani Izsak\thanks{Weizmann Institute of Science; {\tt ran.izsak@weizmann.ac.il}}
\and
Brendan Lucier\thanks{Microsoft Research; {\tt brlucier@microsoft.com}}
\and
Vasilis Syrgkanis \thanks{Cornell University; {\tt vasilis@cs.cornell.edu}}
}
\date{}

\maketitle

\begin{abstract}
We introduce a new hierarchy over monotone set functions, that we refer to as $\MOPH$ (Maximum over Positive Hypergraphs).
Levels of the hierarchy correspond to the degree of complementarity in a given function.
The highest level of the hierarchy, $\MOPH$-$m$ (where $m$ is the total number of items) captures all monotone functions. The lowest level, $\MOPH$-$1$, captures all monotone submodular functions, and more generally, the class of functions known as $\XOS$. Every monotone function that has a positive hypergraph representation of rank $k$ (in the sense defined by Abraham, Babaioff, Dughmi and Roughgarden [EC 2012]) is in $\MOPH$-$k$. Every monotone function that has supermodular degree $k$ (in the sense defined by Feige and Izsak [ITCS 2013]) is in $\MOPH$-$(k+1)$. In both cases, the converse direction does not hold, even in an approximate sense.  We present additional results that demonstrate the expressiveness power of $\MOPH$-$k$.

One can obtain good approximation ratios for some natural optimization problems, provided that functions are required to lie in low levels of the $\MOPH$ hierarchy. We present two such applications. One shows that the maximum welfare problem can be approximated within a ratio of $k+1$ if all players hold valuation functions in $\MOPH$-$k$.
The other is an upper bound of $2k$ on the price of anarchy of simultaneous first price auctions.

Being in $\MOPH$-$k$ can be shown to involve two requirements -- one is monotonicity and the other is a certain requirement that we refer to as $PLE$ (Positive Lower Envelope). Removing the monotonicity requirement, one obtains the $\PLE$ hierarchy over all non-negative set functions (whether monotone or not), which can be fertile ground for further research.

\end{abstract}

%
%
%
%
%
%

\newpage

\pagenumbering{arabic}
\setcounter{page}{1}

\section{Introduction}

In a combinatorial auction setting, a set $M$ of $m$ items is to be allocated among a set $N$ of $n$ buyers. Each buyer $i\in N$ has a valuation function that assigns a non-negative real number $v_i(S)$ to every bundle of items $S\subseteq M$ . A well motivated objective is to find a partition of the items $X=(X_1, \ldots, X_n)$ among the buyers so as to maximize the {\em social welfare}, defined as the sum of buyers' valuations from the bundles they obtain $SW(X)=\sum_{i\in N}v_i(X_i)$. The model of combinatorial auctions is highly applicable to real-world settings such as spectrum auctions and electronic advertisement markets.

Most of the existing literature on combinatorial auctions has focused on the case where buyer valuations are complement-free.
Roughly speaking, this means that the value for the union of two bundles of items cannot exceed the sum of the values for each individual bundle.
Such valuations do not capture scenarios where certain items produce more value when acquired in conjunction with each other (such as a left and right shoe).
Complement-free valuations are arguably more well-behaved than general combinatorial valuations in many aspects.
From an algorithmic perspective, complement-free valuations admit constant-factor polynomial time approximation algorithms \cite{FV,DS2006,Feige2006}, while general valuations are hard to approximate even to within a factor that is sub-polynomial in the number of items \cite{LCS99}.
From a game-theoretic perspective, simple auctions, such as running simultaneously a single-item first-price auction for each item, induce equilibria that achieve constant factor approximations to the optimal welfare if all valuations are complement-free \cite{Syrgkanis2013,Feldman2013,Hassidim2011,Christodoulou2008,Bhawalkar2011}. In contrast, if valuations exhibit complements, the worst-case inefficiency grows with the number of items \cite{Hassidim2011}.

While the theory suggests that complementarities degrade the performance of combinatorial auctions, they arise very naturally in many economic scenarios. A prominent example is the FCC spectrum auctions, where it is desirable to win licenses for the same band of spectrum in adjacent geographical regions. The prevalence of complementarities in practice calls for a better theoretical understanding of 
%
the effect of the level of complementarity on the performance of auctions. 

To this end, we introduce a new hierarchy of monotone set functions called {\em maximum over positive hypergraphs} ($\MOPH$), whose level captures the degree of complementarity. 
A new hierarchy is useful if it has a strong expressiveness power on the one hand, and algorithmic and economic implications on the other.
We show that important classes of functions are captured in low levels of our hierarchy (a detailed exposition is deferred to Section \ref{sec:related}.).
%
We then present algorithmic and economic results that illustrate the usefulness of our hierarchy.  In particular, we develop an algorithm that approximates the welfare maximization problem to within a factor of $k+1$, where $k$ is the degree of complementarity of the valuations, as captured by our hierarchy. We further show that an auction that solicits bids on each item separately and allocates each item to the highest bidder (at a cost equals to her bid) achieves a $2k$-approximation to the optimal welfare at any equilibrium of bidder behavior.
\subsection{The Maximum over Positive Hypergraph ($\MOPH$) Hierarchy}\label{sec:moh}\label{sec:prelims}
Given a set $M$ of $m$ items, a set function $v: 2^M \to \nonnegR$ is {\em normalized} if $v(\emptyset) = 0$ and {\em monotone} if $v(T) \ge v(S)$ whenever $S \subseteq T \subseteq M$.\footnote{We use $\nonnegR$ for non-negative real numbers. That is, 0 is included.}
A normalized monotone set function is necessarily non-negative.
Throughout the paper we assume that all set functions are normalized and monotone, unless stated otherwise.
In the context of combinatorial auctions, we refer to the set functions as valuation functions.

A set function $v$ is symmetric if $v(S) = v(T)$ whenever $|S| = |T|$. A hypergraph representation of a set function $v:2^M \to \nonnegR$ is a (normalized but not necessarily monotone) set function $h:2^M \to \mathbb{R}$ that satisfies $v(S) = \sum_{T \subseteq S} h(T)$. It is easy to verify that any set function $v$ admits a unique hypergraph representation and vice versa.
A set $S$ such that $h(S) \neq 0$ is said to be a {\em hyperedge} of $h$.
Pictorially, the hypergraph representation can be thought of as a weighted hypergraph, where every vertex is associated with an item in $M$, and the weight of each hyperedge $e\subseteq M$ is $h(e)$. Then the value of the function for any set $S\subseteq M$, is the total value of all hyperedges that are contained in $S$. 

The {\em rank} of a hypergraph representation $h$ is the largest cardinality of any hyperedge. Similarly, the {\em positive rank} (respectively, {\em negative rank}) of $h$ is the largest cardinality of any hyperedge with strictly positive (respectively, negative) value. The rank of a set function $v$ is the rank of its corresponding hypergraph representation, and we refer to a function $v$ with rank $r$ as a \emph{hypergraph-$r$} function. Last, if the hypegraph representation is non-negative, i.e. for any $S\subseteq M$, $h(S)\geq 0$, then we refer to such a function as a \emph{positive hypergraph-$r$} ($\PH$-$r$) function .

We define a parameterized hierarchy of set functions, with a parameter that corresponds to the degree of complementarity.
\begin{defn}[Maximum Over Positive Hypergraph-$k$ ($\MOPH$-$k$) class]
A monotone set function $v:2^M\to \nonnegR$ is {\em Maximum over Positive Hypergraph-$k$} ($\MOPH$-$k$) if it can be expressed as a maximum over a set of $\PH$-$k$ functions.  That is, there exist $\PH$-$k$ functions $\{v_{\ell}\}_{\ell\in\Ell}$ such that for every set $S \subseteq M$,
\begin{equation}
\textstyle{v(S) = \max_{\ell \in \Ell} v_{\ell}(S)},
\end{equation}
where $\Ell$ is an arbitrary index set.
\end{defn}

The $\MOPH$ hierarchy has the following attributes:

\begin{enumerate}
\setlength{\itemsep}{1pt}
\setlength{\parskip}{0pt}
\setlength{\parsep}{0pt}
\item {\bf Completeness}. Every monotone set function is contained in some level of the hierarchy (see below).

\item {\bf Usefulness}.
The hierarchy has implications that relate the level in the hierarchy to the efficiency of solving optimization problems.
Specifically, we show implications of our hierarchy to the approximation guarantee of the algorithmic welfare maximization problem (in Section~\ref{sec:algorithmic}) and the price of anarchy of simultaneous single item auctions (in Section~\ref{sec:poa}).

\item {\bf Expressiveness}. The hierarchy is expressive enough to contain many functions in its lowest levels (see Section~\ref{sec:results}).
\end{enumerate}

We conclude this section with some basic properties of the $\MOPH$ hierarchy (for more properties, see Section~\ref{sec:results}).
The two extreme cases of $\MOPH$-$k$ functions coincide with two important classes of valuations.
Specifically, $\MOPH$-$1$ is the class of functions that can be expressed as the maximum over a set of additive functions. This is exactly the class of $\XOS$ valuations~\cite{Lehmann2001}, which is a complement-free valuation class that has been well-studied in the literature. This class contains all submodular valuations, i.e. valuations that exhibit decreasing marginal returns. On the other side, $\MOPH$-$m$ coincides with the class of all monotone functions,\footnote{Simply create a separate $\PH$-$|S|$ function for each set $S$ with a single hyperedge equal to the set $S$ and with weight $f(S)$. Then, by monotonicity, the maximum of these functions is equal to the initial valuation.} and so the hierarchy is complete. For intermediate values of $k$, $\MOPH$-$k$ is monotone; namely, for every $k < k'$ it holds that $\MOPH$-$k\subset \MOPH$-$k'$.
We get the following hierarchy:
\begin{equation} \label{eq:moph_hierarchy}
\text{Submodular}\subset XOS = \MOPH\text{-}1 \subset \cdots \subset \MOPH\text{-}m = \text{Monotone}
\end{equation}

\begin{figure}[tpb]
\label{fig:hypergraph}
\centering
\input{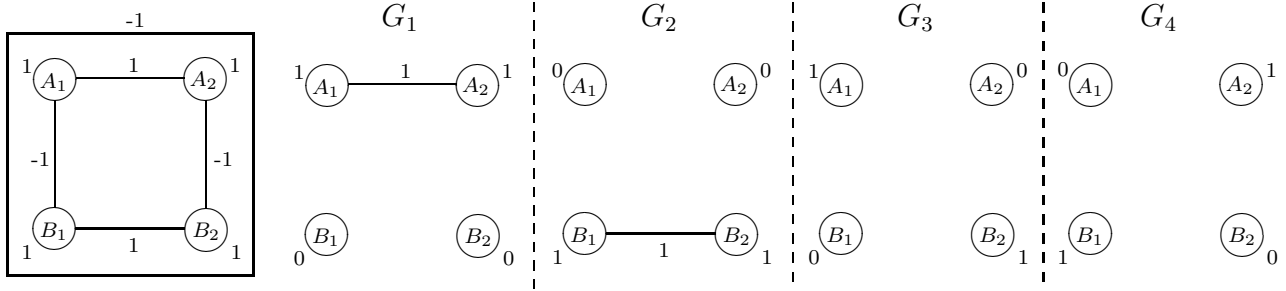}
\caption{The left figure depicts a spectrum auction inspired hypergraph valuation with positive edges and negative hyperedges, which can be expressed as the maximum over the positive graphical valuations on the right.}
\end{figure}

\paragraph{A simple example.}
Consider the example depicted in Figure \ref{fig:hypergraph}, which has an intuitive interpretation in the context of FCC spectrum auctions.
Suppose that $A$, $B$ are two spectrum bands
and that $A_i,B_i$ are auctions representing band $A$ or $B$ at location $i$. Locations $1$ and $2$ are
neighboring geographic regions and therefore, a bidder gets a much larger value for getting the same band
in both regions. Therefore, $A_1$ and $A_2$ have a complementary relationship and similarly $B_1$ and $B_2$. However,
each $A_i$ has a substitute relationship with $B_i$ and additionally the pair $(A_1,A_2)$ has a substitute relationship with the pair $(B_1,B_2)$, since a bidder will only utilize one pair of bands. This valuation can be represented as a hypergraph, as in the left-most diagram in Figure~\ref{fig:hypergraph}. Also, as illustrated in Figure~\ref{fig:hypergraph}, this valuation can be represented as a maximum over positive hypergraph valuations of rank $2$.

\paragraph{Fractionally ``Subadditive'' Characterization of $\MOPH$-$k$.}
We show that the definition of $\MOPH$-$k$ functions has a natural analogue as an extension of {\it fractionally subadditive} functions. See Appendix~\ref{sec:fractionally-subadditive}.

\subsection{Related Work}
\label{sec:related}

\Omit{\vsdelete{Our motivation for introducing the $\MOPH$ hierarchy comes from the study of combinatorial auctions. However, we remark that some classes of set functions received much attention in combinatorial optimization even prior to the interest in combinatorial auctions. This is especially true for the class of {\em submodular} functions on which there is a vast body of literature. We further remark that outside the context of combinatorial auctions, set functions are often not required to be monotone, and hence questions such as submodular function minimization (which can be solved in polynomial time, see~\cite{Cunningham}, for example) and maximization (which can be approximated within a factor of~$1/2$, see~\cite{BFNS}) are natural in these contexts. Similar questions may be relevant to our $\PLE$ hierarchy, which is a variation on the $\MOPH$ hierarchy that does not require monotonicity.}}

\paragraph{Expressiveness.} 
Since the maximum welfare allocation problem is $\NP$-hard to approximate even with very poor ratio (see for example~\cite{LCS99} for the case of {\em single minded bidders} -- bidders that want one particular bundle of items), there has been extensive work on classification of monotone set functions. We distinguish between two types of classifications. One is that of {\em restricted classes} of set functions, and the other is {\em inclusive hierarchies} that capture all monotone set functions.

{\bf Restricted classes of monotone set functions.}
Lehmann, Lehmann and Nisan~\cite{Lehmann2001} initiated a systematic classification of set functions without complementarities. The main classes in their classification (in order of increasing expressiveness) are 
{\em additive}, {\em gross substitutes} (a class introduced by Kelso and Crawford~\cite{KC82}), {\em submodular}, {\em $\XOS$} (the terminology for this class is taken from earlier work of Sandholm~\cite{Sandholm}), and 
{\em subadditive}. 
Subsequent research showed that the maximum welfare can be approximated within a ratio somewhat better than $1 - 1/e$ in the submodular case~\cite{FV}, $1 - 1/e$ in the $\XOS$ case~\cite{DS2006} and~2 in the subadditive case~\cite{Feige2006}. These approximation algorithms assume {\em demand queries} access to the valuation functions, though for the submodular case, if one is satisfied with a $1 - 1/e$ ratio, then {\em value queries} suffice~\cite{Vondrak}. It follows from the definitions that $\MOPH$-1, the first level of our hierarchy, coincides with the class of $\XOS$ functions.

Conitzer, Sandholm, and Santi~\cite{Conitzer2005}
consider the class of 
{\em graphical valuations}. Namely, every item has a weight, and every pair of items (edge of the graph) has a weight (positive if the items are complements, negative if they are substitutes, and~0 if they are independent), and the value of a set of items is the sum of weights of items and edges within the set. It will follow (though this is a nontrivial claim that requires a proof) that 
this class is contained in $\MOPH$-2.

Abraham, Babaioff, Dughmi, and Roughgarden~\cite{Abraham2012} consider the hierarchy of $\PH$-$k$ valuation functions, as already defined, (which are obviously contained in $\MOPH$-$k$) that allows only complements but no substitutes. In particular, submodular functions cannot be expressed in this hierarchy, and moreover, even some supermodular functions cannot be expressed in this hierarchy. 
It is shown in~\cite{Abraham2012} that the maximum welfare problem can be approximated within a ratio of~$k$ if all valuation functions are in $\PH$-$k$. They also discuss mechanisms that are truthful in expectation, but the approximation ratios achieved by their mechanisms deteriorate with the total number $m$ of items, even if $k$ remains fixed. 

{\bf Inclusive hierarchies of monotone set functions.} Feige and Izsak~\cite{Feige2013} introduced a hierarchy of monotone set functions, parameterized by the so-called {\em supermodular degree}. Unlike the $\MOPH$ hierarchy whose levels are numbered from~1 to~$m$, the levels of the supermodular degree hierarchy are numbered from~0 to $m-1$, and one should keep this in mind when comparing levels of these hierarchies. The functions with supermodular degree~0 are the submodular functions. As shown in~\cite{Feige2013}, there is a greedy algorithm that approximates the welfare maximization problem within a ratio of $k+2$ if the supermodular degree of all valuation functions is at most $k$.
We show 
that for every $k$, functions of supermodular degree $k$ are in $\MOPH$-$(k+1)$. For the other direction, there are functions in $\MOPH$-2 that cannot even be approximated by functions of low supermodular degree (e.g., functions of supermodular degree $\sqrt{m}$ approximate them only within a ratio of $\Omega(\sqrt{m})$).

The $\XOS$ class introduced in~\cite{Lehmann2001} is based on ``OR" and ``XOR" operations previously introduced in~\cite{Sandholm}, but with the restriction that ``OR" operations are applied on single items. Removing this restriction and allowing operations on bundles, one obtains an $\XOS$ hierarchy parameterized by the size of the largest bundle. (The $\XOS$ hierarchy was suggested to us in personal communication by Noam Nisan.) While $\XOS$-$k$ and $\MOPH$-$k$ coincide for $k=1$, $\MOPH$-$k$ is strictly better than $\XOS$-$k$. It can be shown that $\XOS$-$k$ is contained in $\MOPH$-$k$, whereas there are functions in $\MOPH$-2 that cannot even be approximated in $\XOS$-$k$ for any constant $k$. 
(The proof uses the $\PH$-2 function used in Section~\ref{app:incompatible}).

\paragraph{Welfare approximation}


The complement-free valuations introduced in \cite{Lehmann2001} have also been studied in the game-theoretic context of equilibria in simultaneous single-item auctions.  It has been established that the Bayes-Nash and Correlated price of anarchy of this auction format, with a first-price payment rule, are at most $e/(e-1)$ in the $\XOS$ case~\cite{Syrgkanis2012a} and at most~2 in the subadditive case~\cite{Feldman2013}. For the second-price payment rule, these bounds become 2 for $\XOS$~\cite{Christodoulou2008} and $4$ for subadditive~\cite{Feldman2013}. These results build upon a line of work studying non-truthful item auctions for complement-free valuations \cite{Bikhchandani1999,Bhawalkar2011,Christodoulou2008,Hassidim2011,PaesLeme2012}.

Equilibrium analysis of non-truthful auctions has been applied to several other settings such as position auctions \cite{Caragiannis2012}, bandwidth allocation mechanisms \cite{Syrgkanis2013}, combinatorial auctions with greedy allocation rules \cite{Lucier2010}, and multi-unit auctions \cite{Markakis2012}. Recently, Roughgarden \cite{Roughgarden2012}, Syrgkanis \cite{Syrgkanis2012} and Syrgkanis and Tardos \cite{Syrgkanis2013} proposed general smoothness frameworks for bounding the social welfare obtained in equilibria of (compositions of) mechanisms. 

\section{Summary of Results} \label{sec:results}

We obtain results on the expressiveness power of the $\MOPH$ hierarchy and show applications of it for approximating social welfare in combinatorial auctions. 
%
\paragraph{Expressiveness.}
The first theorem establishes the expressiveness power of $\MOPH$. For some limitations, see Appendix~\ref{sec:limitations}.

\begin{theorem} \label{t:expressiveness} \label{thm:expressiveness}
The $\MOPH$ hierarchy captures many existing hierarchies, as follows:
\begin{enumerate}
\setlength{\itemsep}{2pt}
\setlength{\parskip}{0pt}
\setlength{\parsep}{0pt}
\item By definition, $\MOPH$-$1$ is equivalent to the class $\XOS$ (defined by Lehmann, Lehmann and Nisan \cite{Lehmann2001}) and every function that has a positive hypergraph representation of rank $k$ (defined by Abraham et al. \cite{Abraham2012}, see Section~\ref{sec:prelims}) is in $\MOPH$-$k$.

\item Every monotone graphical valuation (defined by Conitzer et al. \cite{Conitzer2005}) is in $\MOPH$-2. Furthermore, every monotone function with positive rank~2 is $\MOPH$-2 (see Sections~\ref{sec:flow-proof} and \ref{sec:neg-hyperedge-proof}).



\item Every monotone function that has a hypergraph representation with positive rank $k$ and laminar negative hyperedges (with arbitrary rank) is in $\MOPH$-$k$ (See Section~\ref{sec:laminar-proof}).


\item Every monotone function that has supermodular degree $k$ (defined by Feige and Izsak \cite{Feige2013}) is in $\MOPH$-$(k+1)$ (See Section~\ref{sec:supermodular}).

\end{enumerate}
\end{theorem}

We further establish that the converse direction does not hold, even in an approximate sense, and conclude that
the $\MOPH$-$k$ hierarchy is \emph{strictly} more expressive than many existing hierarchies.
Specifically, we show that $\MOPH$-$1$ and $\MOPH$-$2$ contain functions that cannot be \emph{approximated} by functions in low levels of other hierarchies.
We first state the notion of approximation and then the proposition, whose proof is deferred to Appendix~\ref{app:incompatible}.

\begin{defn} \label{def:approximate}
\label{def:approximation}
We say that a set function $f$ approximates a set function $g$ within a ratio of $\rho \ge 1$ if there are $\rho_1$ and $\rho_2$ such that for every set $S$
$\rho_1 \le \frac{f(S)}{g(S)} \le \rho_2$, and $\frac{\rho_2}{\rho_1} \le \rho$.
\end{defn}




\begin{prop} \label{prop:incompatible}
There are functions in very low levels of the $\MOPH$ hierarchy that cannot be approximated well even at relatively high levels of other hierarchies, as follows:
\begin{enumerate}
\setlength{\itemsep}{2pt}
\setlength{\parskip}{0pt}
\setlength{\parsep}{0pt}
\item There exists a submodular function (i.e., supermodular degree $0$, $\MOPH$-$1$) such that
\begin{enumerate}
\setlength{\itemsep}{1pt}
\setlength{\parskip}{0pt}
\setlength{\parsep}{0pt}
\item A graphical function cannot approximate it within a ratio better than $\Omega(m)$.
\item A positive hypergraph function cannot approximate it within a ratio better than $m$.
\item A hypergraph function of rank $k$ (both negative and positive) cannot approximate it within a ratio better than $\Omega(\frac{m}{k^2})$, for every $k$.
\end{enumerate}
\item There exists a $\PH$-2 function (i.e., $\MOPH$-2) such that every function of supermodular degree $d$ cannot approximate it within a ratio better than $\Omega(m/d)$.
\end{enumerate}
\end{prop}

\paragraph{Applications.}
With the new hierarchy at hand, we are in a position to revisit fundamental algorithmic and game-theoretic problems about welfare maximization in combinatorial auctions.
Our results are reassuring: we obtain good approximation ratios for settings with valuations that lie in low levels of the $\MOPH$-$k$ hierarchy.
From the algorithmic perspective, we provide a polynomial time $(k+1)$-approximation algorithm for the welfare maximization problem when valuations are $\MOPH$-$k$ (assuming access to demand oracles).

\begin{theorem} \label{t:algorithmic}
If all agents have $\MOPH$-$k$ valuations, then there exists an algorithm that gives $k+1$ approximation to the optimal social welfare.
This algorithm runs in polynomial time given an access to demand oracles for the valuations.
\end{theorem}

Our approximation algorithm first solves the configuration linear program for welfare maximization introduced by \cite{DNS}.
As is well known, solving this LP can be done in polynomial time using demand queries.
We then round the solution to the LP so as to get an integer solution. Our rounding technique is oblivious and does not require access to demand queries.
By analyzing the integrality gap, it is established that our rounding technique is nearly best possible.

The second setting we consider is a simultaneous first-price auction --- where each of the $m$ items is sold via a separate single-item auction.
We quantify the welfare loss in this simple auction when bidders have $\MOPH$-$k$ valuations.
We find that the price of anarchy is at most $2k$.

\begin{theorem} \label{t:poa}
For simultaneous first price auctions, when bidders have $\MOPH$-$k$ valuations, both the correlated price of anarchy and the Bayes-Nash price of anarchy are at most $2k$.
\end{theorem}

Our proof technique extends the analysis for complement-free valuations in \cite{Feldman2013} and the smoothness framework introduced in \cite{Syrgkanis2013} to settings with complementarities.
We also establish an almost matching lower bound (See section \ref{sec:poa-lb}).

\begin{theorem} \label{thm:poa-lb}
There exists an instance of a simultaneous first price auction with single minded bidders in $\MOPH$-$k$ in which the price of anarchy is $\Omega(k)$.
\end{theorem}

\paragraph{Remarks.}
Most of our expressiveness results showing that a certain function belongs to $\MOPH$-$k$ are established by showing that the function satisfies a certain requirement that we refer to as the Positive Lower Envelope (PLE) condition.
We also observe that, together with monotonicity, this requirement becomes a sufficient and necessary condition for membership in $\MOPH$-$k$.
This observation motivates the definition of a new hierarchy, referred to as $\PLE$.  The class $\PLE$-$k$ contains $\MOPH$-$k$, but also includes non-monotone functions.
While monotonicity is a standard assumption in the context of combinatorial auctions, $\PLE$ can be applicable outside the scope of combinatorial auctions.
We defer to Appendix \ref{sec:non-monotone}, an analysis of the expressiveness of $\PLE$ functions and the observation that our approximation results extend to non-monotone $\PLE$ functions.

\paragraph{Extensions.}
One of the main open problems suggested by this work is the relation between hypergraph valuations of rank $k$ and $\MOPH$-$k$ valuations. We make the following conjecture:
\begin{conj} \label{conj:quadratic}
Every hypergraph function with rank $k$ (positive or negative) is in $\MOPH$-$O(k^2)$.
\end{conj}
We make partial progress toward the proof of this conjecture,  by confirming it for the case of symmetric functions. For non-symmetric, observe that for the case of laminar negative hyperedges, we show an even stronger statement in item (3) of Theorem \ref{t:expressiveness}.
\begin{theorem} \label{t:extensions}\label{t:PleForAnyMonotoneSymmetricFunction}
Every monotone symmetric hypergraph function with rank $k$ (positive or negative) is in $\MOPH$-$O(k^2)$ (See Sections \ref{sec:symmetric-proof} and \ref{app-symmetric}).
\end{theorem}

For symmetric functions, we conjecture a more precise bound of $\left\lceil \frac{k}{2} \right\rceil \left\lceil \frac{k+1}{2} \right\rceil$, suggested by a computer-aided simulation based on a non-trivial LP formulation.
For the special cases of symmetric functions of ranks $k=3$ and $4$, we show that they are in $\MOPH$-4 and $\MOPH$-6, respectively, and that this is tight.
We use an LP formulation whose optimal solution is the worst symmetric function possible
for a given rank, and its value corresponds to the level of this worst function in the $\MOPH$ hierarchy.
We bound the value of this LP, by using LP duality. (see Section~\ref{app:smallRanks} for proofs)


\section{Proofs} \label{sec:proofs}

In this section we include a part of our proofs. Due to space constraints, we defer the other proofs to the appendix.

\subsection{Positive lower envelope technique}
Proving that a particular set function $f:2^M\to \nonnegR$ can be expressed as $\MOPH$-$k$ requires constructing a set of $\PH$-$k$ valuations that constitutes the index set $\Ell$ over which the maximum is taken.
In what follows we present a canonical way of constructing the set $\Ell$.
The idea is to create a $\PH$-$k$ function for every subset $S$ of the ground set $M$.
The collection of these $\PH$-$k$ functions, one for each subset, constitutes a valid $\MOPH$-$k$ representation if they adhere to the following condition.

\begin{defn} [Positive Lower Envelope (PLE)] \label{def:PLE}
Let $f: 2^M \to \nonnegR$ be a monotone set function.
A positive lower envelope (PLE) of $f$ is any positive hypergraph function $g$ such that:
\begin{compactitem}
\item $g(M) = f(M)$.
\item For any $S \subseteq M$, $g(S) \leq f(S)$. [No overestimate]
\end{compactitem}
\end{defn}

Before presenting the characterization, we need the following definition.
A function $f:2^M \to \nonnegR$ {\em restricted to a subset $S$}, $S \subseteq M$, is a function $f_S: 2^S \subseteq \nonnegR$ with $f_S(S') = f(S')$ for every $S' \subseteq S$.

\begin{prop} [A characterization of $\MOPH$] \label{p:PLEtoMOH} \label{lem:lower-envelope}\label{p:MOHtoPLE}
A function $f$ is in $\MOPH$-$k$ if and only if it is monotone and $f_S$ admits a lower envelope of rank $k$ for every set $S \subseteq M$. 
\end{prop}
%

\subsection{Some proofs of expressiveness}\label{sec:symmetric-proof}\label{sec:flow-proof} 

We provide here a sketch of the proof of the second assertion in Theorem~\ref{thm:expressiveness}, namely that any monotone set function of positive rank~2 is in $\MOPH$-2


\begin{proof}
Let $v:2^M \to \nonnegR$ be a monotone set function of positive rank~2 and let $G_v$ be the hypergraph representation of $v$, where the vertices of $G_v$ are the items of $M$. By Proposition~\ref{lem:lower-envelope} it suffices to show that every $S \subseteq M$ has a positive lower envelope of rank~2 (abbreviated as PLE-2). Consider an arbitrary $S \subseteq M$. We construct a positive lower envelope for $S$ by induction.
Starting with an empty set of vertices, we iteratively add the vertices of $S$, one at a time. Let $u_i \in S$ denote the vertex added at iteration $i$, and $S_i \subseteq S$ denote the resulting subset. The inductive invariant that we maintain is that each $S_i$ has a PLE-2. The base case of the induction is $S_1$, and there the inductive hypothesis holds because $v$ is nonnegative. We now prove the inductive step. Namely, we assume that $S_{i-1}$ has a PLE-2, and prove the same for $S_i$.

Let $N_i$ ($P_i$, respectively) denote the set of negative (positive, respectively) hyperedges in $G_v$ that contain $u_i$ and are contained in $S_i$. (As $v$ has positive rank~2, the hyperedges in $P_i$ have rank at most~2.)
Consider an auxiliary bipartite graph $H$ with members of $N_i$ as one set of vertices, members of $P_i$ as the other set of vertices, and edges between $e \in N_i$ and $e' \in P_i$ iff $e' \subset e$ (namely, the negative hyperedge contains all items of the positive hyperedge). These edges have infinite capacities. Add two auxiliary vertices, $s$ connected to each member of $N_i$ by an edge of capacity equal to the (absolute value of the) weight of the corresponding hyperedge in $G_v$, and $t$ connected to each member of $P_i$ by an edge of capacity equal to the weight of the corresponding hyperedge in $G_v$. We claim that there is a flow $F$ from $s$ to $t$ saturating all edges of $s$. This follows from the max flow min cut theorem, together with the facts that $v$ is monotone and all positive hyperedges have rank at most~2. Given this claim (whose proof appears in Appendix~\ref{sec:neg-hyperedge-proof}), we add to the PLE-2 of $S_{i-1}$ only the the members of $P_i$ (hence positive edges of rank at most~2), but each of them with a weight reduced by the amount of flow that goes from it to $t$ (according to the saturated flow $F$). The flow $F$ gives us a way of charging every negative hyperedge that is discarded against a reduction in weight of positive hyperedges contained in it, and this implies (see details in Appendix~\ref{sec:neg-hyperedge-proof}) that the result is indeed a PLE-2 for $S_i$.
\end{proof}


Next we provide a sketch of the proof of Theorem \ref{t:extensions}, namely that every monotone symmetric hypergraph-$r$ function is in $\MOPH$-$O(r^2)$.

\begin{proof}
Let $f$ be a normalized monotone symmetric set function of rank $r$, and let $h$ be its hypergraph representation. Consider the following normalized monotone symmetric set function $g$ defined by its positive hypergraph representation $p$: $p(S) = f(U)/{n \choose R}$ if $|S| = R$, and $p(S) = 0$ otherwise, for $R=3r^2$. As all four functions $f,h,g,p$ are symmetric, we shall change notation and replace $f(S)$ by $f(|S|)$. As special cases of this notation, $f(U)$ is replaced by $f(n)$, and $f(\phi)$ is replaced by $f(0)$.

We claim that $g$ is a lower envelope for $f$. There are three conditions to check. Two of them trivially hold, namely, $g(0) = f(0) = 0$, and $g(n) = {n\choose R}p(R) = f(n)$. The remaining condition requires that $g(k) \le f(k)$ for every $1 \le k \le n-1$. This trivially holds for $k < R$ because in this case $g(k) = 0$, whereas $f(k) \ge 0$. Hence the main content of our proof is to establish the inequality $g(k) \le f(k)$ for every $R \le k \le n-1$.

The proof proceeds by means of contradiction: suppose there is some $f$ that serves as a negative example, namely, that for this $f$ there is $R \le k \le n-1$ for which $g(k) > f(k)$. We can show that if such an example exists then there exists one where $k=n-1$ (Details appear in Appendix~\ref{app-symmetric}).  Thus it suffices to show that $g(n-1) = {n-1 \choose R}f(n)/{n \choose R} = \frac{n-R}{n}f(n)\leq f(n-1)$ for any $f$ that is hypergraph-$r$.

We will consider the (not necessarily monotone) degree $r$ polynomial $F(x)= \sum_{i=1}^r {x \choose i}h(i)$, that matches $f(x)$ at integral points $\{0,\ldots,n\}$.  Let $M = \max_{0 \le x \le n} |F(x)|$ and let $0 \le y \le n$ be such that $|F(y)| = M$. By Markov's inequality regarding bounds on derivatives of polynomials~\cite{markov}, we can show that $\max_{0 \le x \le n} |F'(x)| \le \frac{2r^2}{n}M$.  If $y$ is an integer then monotonicity of $f$ (and hence of $F$ on integer points) implies that $M = f(n)$. However, $y$ need not be integer. In that case $i < y < i+1$ for some $0 \le i \le n-1$. Let $m = \max[|F(i)|,|F(i+1)|]$. Then
$M \le m + \frac{1}{2}\max_{i \le x \le i+1}[|F'(x)|] 
\le f(n) + \frac{r^2}{n}M$. As $n \ge R \ge 3r^2$ we obtain that $M \le 3f(n)/2$. On the other hand, $f(n-1)=F(n-1)\ge f(n) - \max_{0 \le x \le n}F'(x) \ge f(n) - \frac{2r^2}{n}M \ge f(n) - \frac{3r^2}{n}f(n)$. Since $R = 3r^2$ we have that $f(n - 1) \ge (1 - \frac{R}{n})f(n) = g(n-1)$, as desired.
\end{proof}

\subsection{Algorithmic Welfare Maximization (Proof of Theorem~\ref{t:algorithmic})}
\label{sec:algorithmic}

In this section we consider the purely algorithmic problem, ignoring incentive constraints.
While constant factor approximations exist for welfare maximization in the absence of complementarities (see~\cite{DNS,Feige2006}), it is not hard to see that complementarities can make the welfare problem as hard as independent set and hence inapproximable to within an almost linear factor. Our hierarchy offers a linear degradation of the approximation as a function of the degree of complementarity.
At a high level, our algorithm works as follows: define the {\it configuration linear program} (LP) (introduced in \cite{DNS}) by introducing a variable $x_{i,S}$ for every agent $i$ and subset of items $S$. Given the valuation function $v_i$ of each agent $i$, the {\it configuration LP} is:
\begin{align}\label{LP:configuration}
\text{maximize}\quad \sum\limits_{i,S} x_{i,S} \cdot v_i(S)& \\\nonumber
\text{s.t.} \quad\quad\sum\limits_{S} x_{i,S} \leq 1 \quad&\forall i\in N \\\nonumber
\sum\limits_{i,S \mid j \in S} x_{i,S} \leq 1\quad&\forall j\in M \text{ and } 
x_{i,S} \geq 0 \quad\forall i\in N, S\subseteq M \nonumber
\end{align}
The first set of constraints guarantees that no agent is allocated more than one set and the second set of constraints guarantees that no item belongs to more than one set.  This LP provides an upper bound on the optimal welfare.  To find a solution that approximates the optimal welfare, we first solve this LP (through duality using {\it demand queries}\footnote{For definition of 
demand queries see Appendix~\ref{app:queries} and for discussion on representation of set functions Appendix~\ref{app:poly-representation}.})
and then round it (see below).
\Omit{
\paragraph{Solving the LP.} The configuration LP has exponentially many variables but only polynomially many constraints. As shown in~\cite{NS2006}, it can be solved through duality using {\it demand queries}\footnote{We recall the definitions of value and demand queries in Appendix~\ref{app:queries}.}.
For further discussion regarding representation of set functions, see Appendix~\ref{app:poly-representation}.
}
\paragraph{Rounding the LP.} The rounding proceeds in two steps.  First each agent $i$ is assigned a tentative set $S'_i$ according to the probability distribution induced by the variables $x_{i,S}$.  Note that this tentative allocation has the same expected welfare as the LP.  However, it may be infeasible as agents' sets might overlap.  We must resolve these contentions. Several approaches for doing this when there are no complementarities were proposed and analyzed in~\cite{DNS,Feige2006}.
However, these approaches will fail badly in our setting, due to the existence of complementarities.
Instead, we resolve contention 
using the following technique:
We generate a uniformly random permutation $\pi$ over the agents and then at each step $t$ for $1\leq t\leq n$, assign agent $i=\pi(t)$ items $S_i=S'_i\setminus\{\cup_{i'=\pi(1)}^{\pi(t-1)}S_{i'}\}$, i.e., those items in his tentative set that have not already been allocated.

The following proposition shows that this way of contention resolution guarantees a loss of at most a factor of $k+1$, when all agents have $\MOPH$-$k$ valuations.

\begin{prop}
\label{p:rounding}
If all agents have $\MOPH$-$k$ valuations, then
given a solution to the configuration LP, the above random permutation rounding algorithm
produces (in expectation) an allocation that approximates the maximum welfare within a ratio no worse than $k+1$.
\end{prop}

\begin{proof} 
First, note that the solution is feasible, since every item is allocated at most once.
We upper bound the approximation guarantee.
The sum of values of tentative sets preserve, in expectation, the value of the optimal welfare
returned by the  configuration LP. Consider an arbitrary agent and his tentative set $T$. This set attained its value according to some positive hypergraph $H$ that has no edges of rank larger than $k$. Consider an arbitrary edge of $H$ contained in $T$, and let $r\leq k$ be its rank. We claim that its expected contribution (expectation taken over the random choices of the other agents and the random permutation) towards the final welfare is at least $1/(r+1)$ of its value. The expected number of other agents who compete on items from this edge is at most $r$ (by summing up the fractional values of sets that contain items from this edge). Given that there are $\ell$ other competing agents, the agent gets all items from the edge with probability exactly $1/(\ell + 1)$. As the expectation of $\ell$ is at most $r$, the expectation of $1/(\ell+1)$ is at least $1/(r+1)$ (by convexity) and hence at least $1/(k+1)$ as the valuation function is $\MOPH$-$k$. The proof follows from linearity of expectation.
\end{proof}

It is known that there is an integrality gap of $k-1+\frac{1}{k}$ for hypergraph matching in $k$-uniform hypergraphs (see Chan and Lau~\cite{CL12} and references therein). These instances are special cases of welfare maximization with $\MOPH$-$k$ valuations.
Hence, our rounding technique in Proposition~\ref{p:rounding} is nearly best possible.
For completeness, we show this integrality gap for our setting in Appendix~\ref{app:integralityGap}.
Recall also that even for the case of single-minded bidders with sets of size up to $k$, it is $\NP$-hard to approximate the welfare maximization problem to a better factor than $\Omega(\frac{\ln k}{k})$.\footnote{This hardness is obtained by an approximation preserving reduction from $k$-set packing given in~\cite{LCS99}, together with a hardness result of \cite{HSS}.}

\subsection{Welfare Maximization at Equilibrium (Proof of Theorem~\ref{t:poa})}
\label{sec:poa}

In this section we study welfare guarantees at equilibrium of the simultaneous item auction, when all agents have $\MOPH$-$k$ valuations.
In a simultaneous item (first-price) auction, every bidder $i \in [n]$ simultaneously submits a bid $b_{ij} \geq 0$ for every item $j \in [m]$.
We write $b_i=(b_{i1},\ldots,b_{im})$ for the vector of bids of bidder $i$, and $b=(b_1, \ldots, b_n)$ for the bid profile of all bidders.
Every item is allocated to the bidder who submits the highest bid on it (breaking ties arbitrarily), and the winning bidder pays his bid.
We let $X_i(b)$ denote the bundle allocated to bidder $i$ under bid profile $b$, and we write $X(b)=(X_1(b),\ldots,X_n(b))$ for the allocation vector under bids $b$. When clear in the context, we omit $b$ and write $X$ for the allocation.
A bidders's utility is assumed to be quasi-linear; i.e., $u_i(b;v_i)  = v_i(X_i(b))-\sum_{j\in X_i(b)}b_{ij}$.
Given a valuation profile $\vals$, the welfare of an allocation $X$ is the sum of the agents' valuations: $SW(X;\vals)=\sum_{i=1}^{n}\val_i(X_i)$.
We also denote the welfare-maximizing allocation for values $\vals$ by $X^*(\vals)$ (or $X^*$ in short), and its social welfare by $\opt(\vals)$.


In this part we assume that the valuations of the players are common knowledge; see end of the section for incomplete information extensions.
A Nash equilibrium is a profile of (possibly random) bids $B=(B_1,\ldots,B_n)$, such that no player's utility can increase by deviating to some other bid.
To quantify the inefficiency that can arise in a simultaneous item auction, we will use the \emph{price of anarchy} (PoA) measure, which is the maximum ratio (over all valuation profiles) of the optimal welfare over the welfare obtained at any Nash equilibrium.
\begin{equation}
POA = \max_{\vals\ ; \ B \colon B \text{~is mixed NE}} \frac{\opt(\vals)}{\E_{\rbid\sim B}[SW(X(\rbid))]}.
\end{equation}

\paragraph{Bounding the PoA.} We provide a proof that the PoA of the auction is at most $4k$, when bidders have $\MOPH$-$k$ valuations. Let $B$ be a randomized bid profile that constitutes a Nash equilibrium under valuations $\vals$. For each item $j\in [m]$, let $P_j=\max_{j} B_{ij}$ be the price of item $j$; $P_j$ is a random variable induced by the bid profile.
Consider what would happen if bidder $i$ deviated from $B$ and instead bid
$b_{ij}^*=2k\cdot \E[P_j]$ on all the items $j\in X_i^*$ and $0$ on the other items. By Markov's inequality bidder $i$ wins each item $j\in X_i^*$ with probability at least $1-\frac{1}{2k}$.
Let $v_i^*$ be the $\PH$-$k$ lower envelope with respect to set $X_i^*$ (recall bidders have $\MOPH$-$k$ valuations).
Then, $v_i(X_i^*)=v_i^*(X_i^*)$ and, for any $X_i \subseteq X_i^*$, $v_i(X_i)\geq v_i^*(X_i)$.
Since $v_i^*$ is a $\PH$-$k$ valuation, each hyperedge of $v_i^*$ has size at most $k$; it then follows byu the union bound that bidder $i$ wins all items in any such hyperedge with probability at least $\frac{1}{2}$. Therefore, the value that the player derives from this deviation is at least $\frac{1}{2}v_i^*(X_i^*)=\frac{1}{2}v_i(X_i^*)$. Hence, his utility from the deviation is at least $\frac{1}{2}v_i(X_i^*) - 2k \cdot\sum_{j\in X_i^*} \E[P_j]$. By the Nash condition his utility at equilibrium is at least this high.

By summing the above bound over all bidders $i$, the sum of bidders' utilities at equilibrium is at least $\frac{1}{2}\opt(v) - 2k \cdot\sum_{j\in [m]} \E[P_j]$. 
Adding the expression for the total utility, we get:
\begin{equation*}
\E[SW(B;\vals)]-\sum_{j\in [m]} \E[P_j] = \sum_i \E[u_i(B; v_i)] \geq \frac{1}{2}\opt(v) - 2k \cdot\sum_{j\in [m]} \E[P_j].
\end{equation*}
Since every player has the option to drop out of the auction, his expected utility must be non-negative.
Therefore, the expected total payment at equilibrium is bounded above by the welfare at equilibrium.
Substituting this in the above inequality gives that $2k\cdot \E[SW(B;\vals)]\geq \frac{1}{2} \opt(v)$, which establishes an upper bound of $4k$ on the PoA, as desired.

\paragraph{Extensions.} 
In Appendix \ref{sec:smoothness} we provide a tighter bound of $2k$ on the PoA via the smoothness framework \cite{Syrgkanis2013}.
By using the smoothness framework, this bound immediately extends to Bayesian settings, where the valuations of the players are unknown but are drawn from commonly known independent distributions. The bound also extends to outcomes derived from no-regret (learning) behavior in repeated games.

\bigskip\noindent

{\bf Acknowledgements.}
We are grateful to Noam Nisan for valuable discussions. We used computer-assisted simulations in order to find feasible solutions for the dual LPs. Specifically, we used Microsoft .NET together with Gurobi (\cite{gurobi}) and IBM CPLEX and also Wolfram Mathematica. We thank Ofer Bartal and Matan Karklinsky for their help in this matter.
The work of Uriel Feige and Rani Izsak is supported in part by the Israel Science Foundation
(grant No. 621/12) and by the I-CORE Program of the Planning and Budgeting Committee and The Israel Science Foundation (grant No. 4/11).
The work of Michal Feldman is supported in part by the European Research Council under the European Union's Seventh Framework Programme (FP7/2007-2013) / ERC grant agreement number 337122.
The work of Vasilis Syrgkanis is partly supported by a Simons Graduate Fellowship in Theoretical Computer Science and part was done while an intern at Microsoft Research.

\bibliographystyle{plain}
\bibliography{complements-bib}

\appendix



\section{Fractional Covers and $\MOPH$ Valuations}\label{sec:fractionally-subadditive}
In this section we show an equivalence between Maximum over Hypergraph Valuations and a generalization of fractionally subadditive valuations,
where fractional covers are defined as covering every possible subset of size at most $k$ of nodes.

More formally, consider a set $S$ of items and let $\ESk{S}$ be all the subsets of $S$ of size at most $k$. 
We say that a collection of sets $\mathcal{T}\subseteq 2^S$ together with a weight $a_T$ for each $T\in \T$ is a fractional cover of all the subsets of size at most $k$ ($k$-fractional cover) of $S$ if $\forall s\in \ESk{S}:\sum_{T\in \T:T\supseteq s}a_T\geq 1$.
A valuation $v:2^M \to \nonnegR$ is $k$-fractionally subadditive if for every $S \subseteq M$ and every $k$-fractional cover $(a_{\mathcal{T}},\mathcal{T})$ of $S$, we have $v(S) \leq \sum_{T\in \T}a_T\cdot v(T)$.

\begin{theorem}
The class of monotone $k$-fractionally subadditive valuations is equivalent to the class of $\MOPH$-$k$ valuations.
\end{theorem}
\begin{proof}
First it is easy to observe that any $\MOPH$-$k$ valuation is $k$-fractionally subadditive:
\begin{align*}
\sum_{T\in\T}a_T\cdot v(T) =~& \sum_{T\in \T}a_T\cdot \max_{\ell \in \Ell}\sum_{s\in \ESk{T}}w_s^{\ell}\geq \max_{\ell\in \Ell} \sum_{T\in \T}a_T \sum_{s\in \ESk{T}}w_s^{\ell}=\max_{\ell\in \Ell} \sum_{s\in \ESk{S}}w_s^{\ell} \sum_{T\in \T:T\supseteq s}a_T\\
\geq~& \max_{\ell\in \Ell} \sum_{s\in \ESk{S}}w_s^{\ell}=v(S)
\end{align*}

To show that any monotone $k$-fractionally subadditive valuation is an $\MOPH$-$k$ valuation, we follow a similar analysis to that carried by Feige \cite{Feige2006}, as follows. 
For every set $S$, we construct a hypergraph-$k$ valuation associated with the set $S$, and denote it by $\ell(S)$.
The set of valuations is then $\Ell=\cup_{S\subseteq [m]}\ell(S)$. 
The hypergraph valuation $\ell(S)$ is constructed such that: (i)
$v(S) = \sum_{s\in \ESk{S}}w_s^{\ell(S)}$, and (ii) for any subset $T\subseteq S: v(T)\geq \sum_{s\in \ESk{T}}w_s^{\ell(S)}$. 
Monotonicity then implies that for any set $S$, $v(S)=\max_{\ell \in \Ell} \sum_{s\in \ESk{S}}w_s^{\ell}$, as desired.

It remains to construct the valuation $\ell(S)$.
To this end, we consider the following linear program and its dual:
\begin{align*}
V(S) =~& \min_{(a_T)_{T\subseteq S}}\sum_{T\subseteq S}a_T\cdot v(T) & C(S) =~& \max_{(w_s)_{s\in \ESk{S}}}\sum_{s\in \ESk{S}}w_s\\
\forall s\in \ESk{S}:~& \sum_{T\supseteq s}a_T\geq 1 & \forall T\subseteq S:~& \sum_{s\in \ESk{T}}w_s\leq v(T)\\
\forall T\subseteq S:~& a_T\geq 0 & \forall s\in \ESk{S}:~& w_s \geq 0
\end{align*}
By definition, every feasible solution to the primal program constitutes a fractional cover of every subset of size at most $k$ of $S$.
Therefore, it follows by $k$-fractional subadditivity that $V(S)\geq v(S)$. Since $v(S)$ can be obtained by setting $a_S=1$ and $a_T=0$ for any $T\subset S$, we get that $V(S)=v(S)$. Duality then implies that $C(S)=v(S)$. Thus if we set $(w_s^{\ell(S)})_{s\in \ESk{S}}$ to be the solution to the dual, then the conditions that need to be hold for $\ell(S)$ are satisfied by the constraints of the dual and the duality.
\end{proof}

\section{$\MOPH$ as a useful representation}
\label{app:poly-representation}
An explicit description (as a truth table) of a set function over $m$ items contains $2^m - 1$ entries. However, one is typically more interested in those set functions that have implicit descriptions of size polynomial in $m$ (which by a simple counting argument constitute only a small fraction of all possible set functions). A set function may have several different polynomial size descriptions, that differ wildly in their usefulness. For example, a useful description of a linear function simply lists the values of all items, whereas a rather useless description is to list $m$ $\NP$-hard optimization problems, with the intended interpretation that the value of the $i$th item is equal to the value of the optimal solution of the $i^{th}$ problem.

Given a description of a set function $f$, a ``minimal" notion of usefulness is that of having the ability to answer {\em value queries}, namely, to have a fast (polynomial time) algorithm that given a set $S$ as input, returns $f(S)$. Such an algorithm is referred to as a {\em value oracle} for $f$. Ideally, the value oracle has a polynomial size implementation, say, as a circuit of size polynomial in $m$. However, in the context of combinatorial auctions, one may also imagine using exponential size implementations of value oracles: an agent may hold a representation of his own valuation function $f$ as a truth table. This representation may be too long to be communicated in its entirety to the welfare maximization algorithm, but the agent may answer value queries quickly upon request, by having random access (rather than sequential access) to this table.

Another notion of usefulness that comes up naturally in the context of combinatorial auctions is that of having the ability to answer {\em demand queries}, namely, to have a fast (polynomial time) algorithm that given a list of item prices $p_i$ as input, returns the set $S$ that maximizes $f(S) - \sum_{i\in S} p_i$. Such an algorithm is referred to as a {\em demand oracle} for $f$. A demand oracle is at least as difficult to implement as a value oracle, because value queries can be implemented using demand queries (see \cite{BN09}).

It is shown in~\cite{FeigeJozeph2014}
that even some submodular functions that have a polynomial representation do not have any polynomial size implementation of demand oracles (not even oracles that answer demand queries approximately), unless $\NP$ has polynomial size circuits. In fact, there are such functions that do not have any polynomial size implementation of value oracles (unless $\NP$ has polynomial size circuits). Hence not all functions in $\MOPH$-1 have useful representations, not even those functions that have a polynomial representation.

A function $f$ is said to have a polynomial size $\MOPH$ representation if it can be represented as a maximum over polynomially many positive hypergraphs, where each hypergraph is of polynomial size. Interestingly, a polynomial size $\MOPH$ representation is always a useful representation.

\begin{prop}
Given a polynomial size $\MOPH$ representation of a set function $f$, one can answer demand queries in polynomial time.
\end{prop}

\begin{proof}
It was observed in \cite{Abraham2012} that one can solve a demand query on a positive hypergraph valuation by the following algorithm: 
subtract the item prices from the corresponding vertices and then observe that the resulting graph (which potentially has negative vertex weights) defines a supermodular function. Supermodular function maximization is equivalent to submodular function minimization, and is well-known to be solvable in strongly polynomial time.
To answer a demand query on a function with a polynomial size $\MOPH$ representation, simply solve the demand query on each hypergraph separately, as proposed above, and then return the maximum.
\end{proof} 

\section{Proofs on Limitations of previous hierarchies}\label{app:incompatible}
In this section we prove Proposition~\ref{prop:incompatible} by constructing functions that lie in $\MOPH$-$k$ for small values of $k$, but that are not approximable by functions that are low in alternative hierarchies.
Recall (from Definition~\ref{def:approximate}) that a set function $f$ approximates a set function $g$ within a ratio of $\rho \ge 1$ if there are $\rho_1$ and $\rho_2$ such that for every set $S$, $\rho_1 \le \frac{f(S)}{g(S)} \le \rho_2$, and $\frac{\rho_2}{\rho_1} \le \rho$.

We begin by proving the first half of Proposition~\ref{prop:incompatible}.  Let $f_1$ be the function such that $f_1(S) = 1$ for every nonempty set $S \subseteq M$.  Note that $f_1$ is submodular.  Recall that the hypergraph rank of a function $f$ is the cardinality of the largest (nonzero) hyperedge in the hypergraph representation of $f$.

\begin{prop}
\label{pro:f_1}
The hypergraph rank of $f_1$ is $|M|$. Moreover, for every $k$ (the following result is nontrivial for $k \le o(\sqrt{m})$), every function of hypergraph rank $k$ approximates $f_1$ within a ratio no better than $\Omega\left(\frac{m}{k^2}\right)$.
\end{prop}

\begin{proof}
The unique hypergraph representation of $f_1$ sets $h(S)=1$ for every odd cardinality set $S$, and $h(S) = -1$ for every nonempty even cardinality set $S$. Hence its hypergraph rank is $|M|$.

For the {\em moreover} part of the proposition, we only sketch the proof; the missing details can be found in the proof of Theorem~\ref{t:PleForAnyMonotoneSymmetricFunction}. Let $f$ be an arbitrary function of hypergraph rank $k$ that approximates $f_1$ within a ratio of $\rho$. We may assume without loss of generality that $f$ is a symmetric function (otherwise we can symmetrize it by averaging all $m!$ functions that are obtained from $f$ by permuting the items of $M$). By scaling we may assume that $f(S) \le 1$ for every set $S$, with equality for some nonempty $S$.  Then the approximation guarantee is that $f(S) \ge 1/\rho$ for every nonempty set $S$. Expressing the symmetric function $f$ as a degree $k$ polynomial in one variable, the derivative of this polynomial is bounded by $O(\frac{k^2}{m})$ (see details in the proof of Theorem~\ref{t:PleForAnyMonotoneSymmetricFunction}). As $f(0) = 0$ this implies that $f(1) \le O(\frac{k^2}{m})$. Combining with the the approximation requirement of $f(1) \ge 1/\rho$ we obtain that $\rho \ge \Omega(\frac{m}{k^2})$, as desired.
\end{proof}

For $k=2$, Proposition~\ref{pro:f_1} immediately implies that graphical valuations cannot approximate $f_1$ within a ratio better than $\Omega(m)$. Likewise, it is similarly seen that no function in the $\PH$ hierarchy can approximate $f_1$ within a ratio better than $m$.  

We next turn to the second half of Proposition~\ref{prop:incompatible}.
Let $f_2$ denote the function with the following graphical representation: the complete graph in which each vertex has value~0 and each edge has value~1. Hence $f(S) = 0$ for $|S| \le 1$ and $f(S) = {|S| \choose 2}$ otherwise. Observe that $f_2$ is a graphical valuation function (by definition) and that $f_2$ is in $\PH$-2 (again, by definition).

\begin{prop}
\label{pro:PHsupermodular}
Every function of supermodular degree $d$ approximates $f_2$ within a ratio no better than $\frac{m}{d+1} - 1$. In particular, $f_2$ has supermodular degree $\Omega(m)$.
\end{prop}

\begin{proof}
Let $f$ be a function of supermodular degree $d$. Then, there must exist a set $T$ with $|T| \geq \frac{m}{d+1}$ such that $f$ restricted to $T$ is submodular (such a set can be obtained by picking an item, discarding the at most $d$ items that have supermodular dependence with it, and repeating). Without loss of generality, normalize $f$ such that $f(T) = f_2(T) = {|T| \choose 2}$.  Then, by submodularity of $f$ on $T$,
there is a set $T' \subset T$ with $|T| = 2$ such that $f(T') \ge \frac{2}{|T|}f(T) = |T|-1$. But $f_2(T') = 1$, implying that the approximation ratio $\rho$ is at least $|T|-1 \ge \frac{m}{d+1} - 1$.
\end{proof}

\section{Positive lower envelope technique}
\label{app:ple-technique}
\begin{proofof} {Proposition~\ref{p:PLEtoMOH}}
\begin{description}
\item[First direction: Monotone, PLE of rank {\boldmath $k$ $\Rightarrow$ $\MOPH$-$k$}:]
Let $f:2^M \to \nonnegR$ be a monotone set function.
For any $T \subseteq M$, let $g_T$ be a positive lower envelope of rank $k$ of $f$ restricted to $T$.
We argue that $\{g_T\}_{T \subseteq M}$ is an $\MOPH$-$k$ representation of $f$.
Specifically, we show that for every $S \subseteq M$, it holds that $\max_{T \subseteq M} g_T(S) = f(S)$.
Let $S \subseteq M$.
By the first property of Definition~\ref{def:PLE}, it holds that $g_S(S) = f(S)$.
Therefore,
\begin{equation} \label{appPle:eq:nounderestimate}
\max_{T \subseteq M} g_T(S) \geq g_S(S) = f(S)
\end{equation}
Additionally, for any $T \subseteq M$, $g_T(S) = g_T(S \cap T) \leq f(S \cap T) \leq f(S)$,
where the equality follows from the fact that $g_T$ is restricted to $T$; the first inequality follows from the no-overestimate property of Definition~\ref{def:PLE} and the last inequality follows from monotonicity of $f$.
Therefore,
\begin{equation} \label{appPle:eq:nooverestimate}
\max_{T \subseteq M} g_T(S) \leq f(S)
\end{equation}
The first direction follows by Equations~\eqref{appPle:eq:nounderestimate} and~\eqref{appPle:eq:nooverestimate}.

\item[Second direction: {\boldmath $\MOPH$-$k$ $\Rightarrow$} Monotone, PLE of rank {\boldmath $k$}:]
Let $f: 2^M \to \nonnegR$ be a function in $\MOPH$-$k$ and let $\Ell$ be an $\MOPH$ representation of it.
We first prove that $f$ is monotone.
Assume towards contradiction that $f$ is not monotone.
Then, there exist $S' \subset S \subseteq M$ such that $f(S') > f(S)$.
Since $\Ell$ is an $\MOPH$ representation of $f$, there exists a positive hypergraph function $f_h \in \Ell$ such that $f_h(S')=f(S')$.
This means that $f_h(S) \geq f(S') > f(S)$, which implies that $\max_{g \in \Ell} g(S) > f(S)$, deriving a contradiction.
The monotonicity of $f$ follows.
We next show that for every set $S \subseteq M$, $f_S$ admits a positive lower envelope of rank $k$.
Let $S \subseteq M$. 
There exists a positive hypergraph function $f_h \in \Ell$ such that $f_h(S)=f(S)$.
Moreover, no set $S' \subseteq S$ can have value strictly greater than $f(S')$ according to $f_h$, since if it does, this will be a lower bound on the value of $S'$ according to $\Ell$.
Therefore $f_h$ is a positive lower envelope of $f_S$, as desired.
The second direction follows. 
\end{description}
We conclude the proof of Proposition~\ref{p:MOHtoPLE}.
\end{proofof}

\section{Proof of Expressiveness Theorem \ref{t:expressiveness}} \label{app:omitted}

\subsection{Graphical Valuations with Negative Hyperedges}
\label{sec:neg-hyperedge-proof}

We prove the part of Theorem~\ref{thm:expressiveness} dealing with monotone set functions with positive rank $2$.

\begin{theorem}
Any monotone set function of positive rank~2 is in $\MOPH$-2
\end{theorem}

\begin{proof}
Let $f$ be a monotone hypergraph set function with positive rank~2, and let $h$ be its hypergraph representation.
In order to prove the theorem, we invoke the lower envelope technique, given in Proposition~\ref{p:PLEtoMOH}.
Specifically, for every set $S$, we construct a monotone positive lower envelope $w^S$ of rank~2 with a hypergraph representation $g^S$.

Fix a set $S$ of cardinality $r$.
For every hyperedge $e$ that is not contained in $S$ or such that $h(e)=0$, set $g^S(e)=0$.
The rest of the proof deals with hyperedges $e$ that are contained in $S$ and such that $h(e)\neq 0$.

Let $s_1, \ldots, s_r$ be an arbitrary order of the items in $S$.
We construct $g^S$ incrementally, considering one item at a time.
For every $j \leq r$, let $S_{<j}$ be the set $\{s_1, \ldots, s_{j-1}\}$.
For an item $s_j$ and a set of items $T \subseteq S_{<j}$, let $E_j^+(T)$ denote the set of positive hyperedges that are created when adding item $s_j$ to the set $T$; i.e., $E_j^+(T) = \{e \in S : h(e)>0, e \subseteq T \cup \{s_j\}, e \not\subseteq T\}$.
Similarly, let $E_j^-(T)$ denote the set of negative hyperedges that are created when adding item $s_j$ to $T$; i.e.,
$E_j^-(T) = \{e \in S : h(e)<0, e \subseteq T \cup \{s_j\}, e \not\subseteq T\}$.
Specifically, $E_j^+(S_{<j})$ and $E_j^-(S_{<j})$ are the respective new positive and negative hyperedges that are created when $s_j$ is added to $S_{<j}$.

The idea is to charge the negative values of the edges in $E_j^-(S_{<j})$ to positive edges in $E_j^+(S_{<j})$ in such a way that the absolute value of every negative hyperedge is charged completely to positive edges that are subsets of the negative hyperedge.
In particular, for every $e\in E_j^+(S_{<j})$, we compute a charging $c_e \leq h(e)$, satisfying the following two conditions:
\begin{align}
\sum_{e\in E_j^+(S_{<j})} c_e =~& \sum_{e\in E_j^-(S_{<j})} |h(e)| \label{eqn:sat} \\
\forall T\subseteq S_{<j}: \sum_{e\in E_j^+(T)} c_e \geq~& \sum_{e\in E_j^-(T)} |h(e)| \label{eqn:non-vio}
\end{align}
Then, for every $e \subseteq S$ such that $h(e)<0$, we set $g^S(e)=0$, and for every $e \subseteq S$ such that $h(e)>0$, we set $g^S(e)=h(e)-c_e$,

In order to prove the theorem, we need to establish two arguments:
First, we need to show that $g^S$, calculated by a charging scheme that satisfies Equations \eqref{eqn:sat} and \eqref{eqn:non-vio}, is a hypergraph representation of a lower envelope of $v$ with respect to $S$.
Second, we need to argue that there exists a charging scheme satisfying Equations \eqref{eqn:sat} and \eqref{eqn:non-vio}.

We start with the first argument.
Namely, we show that $w^S$, with a hypergraph representation $g^S$, satisfies the two conditions of Lemma~\ref{lem:lower-envelope}.
We first establish the first property.
\begin{align*}
v(S) =~& \sum_{j=1}^{r}v(s_j~|~S_{<_j}) = \sum_{j=1}^{r} \left(\sum_{e\in E^+_j(S_{<j})}h(e)-\sum_{e\in E^-_j(S_{<j})}|h(e)|\right)\\
=~&\sum_{j=1}^{r} \left(\sum_{e\in E_j^+(S_{<j})}h(e)-\sum_{e\in E_j^+(S_{<j})}c_e\right)
=\sum_{j=1}^{r} \left(\sum_{e\in E_j^+(S_{<j})}g^S(e)\right)\\
=~&
\sum_{j=1}^{r}w^S(s_j~|~S_{<j}) = w^S(S) \ ,
\end{align*}
where the first and second equalities are by definitions (marginal set function and hypergraph representation, respectively);
the third is by~\eqref{eqn:sat}; the fourth is by definition of $g^S(e)$;
the fifth is by definitions of $g^S$ and $w^S$ and the last is by definition of marginal set function.
We next establish the no-overestimate property; i.e., for every set $T$, $v(T)\geq w^S(T)$.
\begin{align*}
v(T) =~& \sum_{s_j \in T} v(s_j~|~S_{<j}\cap T) = \sum_{s_j\in T}\left(\sum_{e\in E_j^+(S_{<j}\cap T)}h(e)-\sum_{e\in E_j^-(S_{<j}\cap T)}|h(e)|\right)\\
\geq~&\sum_{s_j\in T} \left(\sum_{e\in E_j^+(S_{<j}\cap T)}h(e)-\sum_{e\in E_j^+(S_{<j}\cap T)}c_e\right)
=\sum_{s_j\in T} \left(\sum_{e\in E_j^+(S_{<j}\cap T)}g^S(e)\right)\\
=~&
\sum_{s_j\in T}w^S(s_j~|~S_{<j}\cap T) = w^S(T) \ ,
\end{align*}
where the inequality is by~\eqref{eqn:non-vio}.

We now prove the latter argument; i.e., that a charging scheme satisfying Equations \eqref{eqn:sat} and \eqref{eqn:non-vio} exists.
We reduce the existence problem to a max-flow min-cut computation on an appropriately defined graph.
For every item $s_j \in S$, we associate a flow graph $G$, with source $s$, target $t$, and two sets of nodes $L$ and $R$ as described next.
To avoid confusion between the hypergraph edges and the edges of the flow graph $G$, we refer to the edges of $G$ as arcs.
For every hyperedge $e\in E^-_j(S_{<j})$ associate a node $e \in L$, and draw an arc from the source $s$ to node $e$ with capacity $|h(e)|$.
Similarly, for every hyperedge $e\in E^+_j(S_{<j})$ associate a node $e \in R$, and draw an arc from node $e$ to the target node $t$ with capacity $h(e)$.
In addition, for every node $v \in L$, we draw an arc with capacity $+\infty$ to every node $u \in R$ such that $e_u \subseteq e_v$, where $e_u$ and $e_v$ are the hyperedges associated with nodes $u$ and $v$, respectively.

We will soon show that the maximum flow of this graph is equal to $\sum_{e\in E^-_j(S_{<j})}|h(e)|$ (i.e., all edges from $s$ to nodes in $L$ are fully saturated).
Then, for every node $v \in R$, we let the charging of the edge $e$ associated with $v$ to be equal to the flow on the edge from $v$ to $t$ in the maximum flow.
We argue that such a charging satisfies Equations \eqref{eqn:sat} and \eqref{eqn:non-vio}.
Equation \eqref{eqn:sat} is clearly satisfied, as the total charging on the positive edges is equal to the total flow, which is the total absolute value of the negative hyperedges.
To establish Equation \eqref{eqn:non-vio}, consider a subset $T\subseteq S_{<j}$.
If a hyperedge $e\in E^-_j(S_{<j})$ is also contained in $E^-_j(T)$, then all of the positive edges that are subsets of $e$ are contained in $E^+_j(T)$.
In addition, by the design of the flow graph, all the negative value of $e$ was charged to positive edges that are contained in it.
Therefore, the total charging associated with the positive edges $E^+_j(T)$ is at least as much as the total absolute value of the negative edges in $E^-_j(T)$, establishing Equation \eqref{eqn:non-vio}.

It now remains to establish a flow with value $\sum_{e\in E^-_j(S_{<j})}|h(e)|$.
We equivalently establish that the minimum $(s,t)$-cut in $G$ equals this value.
Clearly, there exists a cut with this value; namely, the set of the outgoing edges from the source $s$.
It remains to show that every cut has at least this value.
Let $(A,B)$ be an arbitrary $(s,t)$-cut of $G$, such that $s \in A$ and $t \in B$, and let $c(A,B)$ denote its value.
We show that
\begin{equation}
\label{eq:cut-lb}
c(A,B) \geq \sum_{e\in E^-_j(S_{<j})}|h(e)|.
\end{equation}
If there exist two nodes $v,u$ such that $v \in L \cap A$ and $u \in R \cap B$, then the cut $(A,B)$ contains an arc with infinite capacity, so Equation \eqref{eq:cut-lb} holds.
Otherwise, for every node in $L \cap A$, every node $u \in R$ connected to $v$ must be in $A$ too.
It follows that the value of the cut can be expressed as follows:
\begin{equation*}
C(A,B) = \sum_{e\in E^-_j(S_{<j})\cap B}|h(e)| + \sum_{e\in E_j^+(S_{<j})\cap A}h(e),
\end{equation*}
and so we need to show that
\begin{equation*}
\sum_{e\in E^-_j(S_{<j})\cap B}|h(e)| + \sum_{e\in E_j^+(S_{<j})\cap A}h(e) \geq \sum_{e\in E^-_j(S_{<j})}|h(e)|.
\end{equation*}
Substituting $\sum_{e\in E^-_j(S_{<j})}|h(e)| = \sum_{e\in E^-_j(S_{<j}) \cap A}|h(e)| + \sum_{e\in E^-_j(S_{<j}) \cap B}|h(e)|$, it remains to show that
\begin{equation*}
\sum_{e\in E_j^+(S_{<j})\cap A}h(e) \geq \sum_{e\in E^-_j(S_{<j})\cap A}|h(e)|.
\end{equation*}

Let $A^-_j = \cup_{e\in E^-_j(S_{<j})\cap A}\left(e \setminus \{s_j\}\right)$ be the union of items in the negative hyperedges in $A$, excluding $s_j$.
Consider the hypothetical case of adding item $s_j$ to $A^-_j$.
Observe that when adding $s_j$ to $A^-_j$ the absolute value of the negative edges created is at least the sum of all the absolute values of negative hyperedges in $A$.
By the fact that the positive edges are of cardinality at most $2$, the total value of the positive edges created is exactly equal to the union of the positive edges that are subsets of some of the negative hyperedges in $A$.
Note that the fact that all positive edges are of cardinality at most $2$ is crucial here:
if there were positive hyperedges of cardinality greater than $2$, then additional hyperedges could have been created and the argument would break.
The latter union is exactly the set of positive edges in $A$. Thus we get:
\begin{equation}
v(s_j~|~A^-_j) \leq \sum_{e\in E_j^+(S_{<j})\cap A}h(e)-\sum_{e\in E^-_j(S_{<j})\cap A}|h(e)|.
\end{equation}
The monotonicity of $v$ implies that $v(s_j~|~S^-(A))\geq 0$, implying that
$$\sum_{e\in E_j^+(S_{<j})\cap A}h(e)-\sum_{e\in E^-_j(S_{<j})\cap A}|h(e)| \geq 0,$$ as required.
\end{proof}

\subsection{Laminar Negative Hyperedges}
\label{sec:laminar-proof}

We prove the part of Theorem \ref{t:extensions} dealing with monotone hypergraph functions with laminar negative hyperedges.

\begin{theorem} \label{t:laminar_nonneg}
Every monotone hypergraph function with positive rank $r$ and laminar negative hyperedges (with arbitrary rank) is in $\MOPH$-$r$
\end{theorem}


\begin{proof}
Let $f: 2^M \to \nonnegR$ be a non-negative set function with positive rank $r$ and hypergraph representation $w:2^M\to \R$. For this proof we actually only assume that $f$ is non-negative and not the stronger assumption that it is monotone. Under solely the non-negativity assumption we construct a valid positive lower envelope of $f$ of rank $k$. Then this would directly imply that when the function is also monotone, then it is in $\MOPH$-$k$.
Note that we construct a positive lower envelope only for $f$ itself, but the proof applies to any restriction of $f$, as well.

Let $E$ denote the set of hyperedges of $w$ (i.e. $S\subset M$ with $w(S)\neq 0$). We will denote with $F=(M,E,w)$ the actual hypergraph on vertices $M$, hyperedge-set $E$ and hyperedge weights $w$. By our assumption the negative hyperedges of $F$ are laminar. We show how to remove the negative hyperedges of $F$ one by one, without either changing its value for $M$, increasing its value for any $S \subseteq M$ or disobeying its non-negativity. Thus at the end of this construction, the positive hypergraph that will remain will be a valid positive lower envelope of $f$.

Let $e$ be a negative hyperedge of minimum rank and let $S_e \subseteq M$ be the set of items represented by vertices contained in $e$.
From non-negativity $f(S_e) \geq 0$. Therefore, there exist positive hyperedges $E^+$ containing only vertices representing items of $S_e$ with sum of values of at least $|w(e)|$.
We remove $e$ and total value of $|w(e)|$ from the edges in $E^+$ (arbitrarily, without introducing any new negative edges).
Let $F'$ be $F$ after this change and let $f'$ be the function $F'$ represents.
We show $F'$ still represents a non-negative function $f'$.
Assume towards contradiction this is not the case.
Let $S^-$ be a minimal set of vertices representing set of items with negative value by $f'$.
From minimality of $S^-$ it must be that removing any vertex of $S^-$ results in removing at least one negative hyperegde from the induced subgraph.
That is, $S^-$ is a collection of negative edges. Let $e^{-}_1, \ldots, e^{-}_l$ be these negative edges. From laminarity, these negative edges must be disjoint.
Therefore, $f'(S^-) = \sum_1^l f'(\items(e^-_i))$ (where $\items(e^-_i)$ are the items represented by the vertices contained in $e^-_i$). Since this sum is strictly negative, it must be that the value of at least one of the addends is negative.
So, from minimality of $S^-$ it must be that $l=1$. That is, $S^-$ must contain the items of exactly one negative edge (which may contain, of course, other negative or positive edges).
But, since $f(S^-) \geq 0$, it must be that $S^-$ contains at least one of the positive edges of $E^+$ and does not contain $e$. Therefore, $e^{-}_1$ and $e$ must not be disjoint.
This contradicts laminarity of negative edges.
Therefore, it must be that $f'$ is a non-negative function.
Moreover, since $f$ is laminar and since no negative edges have been added, $f'$ is laminar as well.
Theorem~\ref{t:laminar_nonneg} follows by an inductive argument by observing that $F'$ contains less negative edges than $F$.
\end{proof}

\subsection{Supermodular Degree}
\label{sec:supermodular-proof}\label{sec:supermodular}
We prove the part of Theorem~\ref{thm:expressiveness} dealing with the supermodular degree hierarchy.

\begin{theorem}\label{thm:supermodular-degree}
Every function with supermodular degree $k$ is in $\MOPH$-$(k+1)$.
\end{theorem}

\begin{proof}
Let $f:2^M\to \nonnegR$ be in $\SMD$-$k$. Invoking Proposition~\ref{p:PLEtoMOH}, we construct a positive lower envelope $g:2^M\to \nonnegR$ of $f$ that has rank $k+1$. We denote by $w:2^M\to \nonnegR$ the (positive) hypergraph representation of $g$.
Note that we construct a positive lower envelope only for $f$ itself, but the proof applies to any restriction of $f$, as well.

Consider an arbitrary ordering of the items: $M=\{1,\ldots,m\}$. For every item $j\in M$, let $e_j$ be the set of items that contains $j$ and all the items $j' \in M$ such that $j'$ is supermodularly dependent with $j$. By the assumption of the supermodular degree of $k$: $|e_j|\leq k+1$. Let $S_{<j}=\{1,\ldots,j-1\}$. For every $j$, we associate a weight of $f(j~|~S_{<j})$ with hyperedge $e_j$ (i.e., $w(e_j)=f(j~|~S_{<j})$), and every other hyperedge has a weight of zero.
For simplicity of presentation, we will allow for multiple identical hyperedges.
If $e_j$ is identical to some other $e_{j'}$, then we implicitly assume that the weight of the hyperedge is the addition of the two weights.

By construction we have that $f(M) = \sum_{j=1}^{m}f(j~|~S_{<j})=\sum_{j=1}^{m} w(e_j)=\sum_{e\in E_w} w(e)=g(M)$,
where $E_w$ is the set of hyperedges of $w$.
Thus, the first property of Definition~\ref{def:PLE} is satisfied. It remains to show that for every subset $T \subseteq M$, it holds that $f(T)\geq g(T)=\sum_{e\in E_w:e\subseteq T}w(e)$. Observe that it suffices to consider only sets $T$ that are unions of hyperedges with positive weights.
This is because any vertex that is not contained in some hyperedge $e_j$ contributes no value to the summand on the right hand side. Thus we assume that $T=\cup_{j\in J}e_{j}$ for some index set $J\subseteq M$. We need to show that $f(T) \geq \sum_{j\in J} w(e_j)$.  Note that:

\begin{align*}
f(T) =\sum_{j\in T} f(j~|~T\cap S_{<j}) = \sum_{j\in T-J} f(j~|~T\cap S_{<j}) + \sum_{j\in J} f(j~|~T\cap S_{<j})\geq  \sum_{j\in J} f(j~|~T\cap S_{<j}),
\end{align*}
where the last inequality follows by monotonicity (which implies that every summand of the first sum is non-negative).

Now we argue that for every $j \in J$, $f(j~|~T\cap S_{<j})\geq f(j~|~S_{<j})$.
To see this observe that for every $j\in J$, $e_j\subseteq T$, which implies that every item $j'\in M$ that supermodularly depends on $j$ is in $T$. Therefore, $S_{<j}-T$ contains only items that do not supermodularly depend on $j$.
Consequently, adding any element in $S_{<j}-T$ to $T\cap S_{<j}$ can only decrease the marginal contribution of $j$; that is:
$$\forall j'\in S_{<j}-T, ~~~ f(j~|~T\cap S_{<j})\geq  f(j~|~T\cap S_{<j}+\{j'\}).$$
Repeating the above analysis for every item in $S_{<j}-T$ establishes the desired claim. Combining the above, we get:
\begin{align*}
f(T) \geq  \sum_{j\in J} f(j~|~T\cap S_{<j})\geq \sum_{j\in J} f(j~|~S_{<j}) = \sum_{j\in J}w(e_j)=g(T),
\end{align*}
as desired.
\end{proof}

\section{Proof of Symmetric Functions Theorem \ref{t:extensions}}\label{app-symmetric}
\begin{defn} [Symmetric set function]
We say that a set function $f: 2^M \to \nonnegR$ is {\em symmetric}, if there exists a function $f': \{0, \ldots, |M|\} \to \nonnegR$ such that for any $S \subseteq M$, $f'(|S|) = f(S)$.
For simplicity, we sometimes refer to $f$ itself as a function getting only the cardinality of a subset and not a subset. 
\end{defn}

\begin{observation}
The hypergraph representation of any symmetric set function is symmetric, in the sense it has the same value for any hyperedge of a given rank.
\end{observation}

\begin{lemma} [Canonical PLE for symmetric set functions] \label{l:canonicalPLE_Symmetric}
Let $f: \{0, \ldots, m\} \to \nonnegR$ be a monotone symmetric set function of rank $r$.
Then, if $f$ has a positive lower envelope of rank $R$, it must be that the symmetric set function $g_f: \{0, \ldots, m\} \to \nonnegR$
that has hypergraph representation with positive hyperedges of rank $R$ of value $\binom{m}{R}f(m)$ and no other (non-zero) hyperedges is a positive lower envelope of $f$.
We refer to $g_f$ as \emph{the canonical positive lower envelope of rank $R$ of $f$}.
\end{lemma}

\begin{proof} 
Let $M$ be a ground set and let $g: 2^M \to \nonnegR$ be a positive lower envelope of $f$ of rank $R$.
We consider $G=(V,E,w)$, the hypergraph representation of $g$.
Note that $G$ must have only positive hyperedges, from the definition of positive lower envelope.
We cancel the hyperedges of $G$ of size less than $R$ as follows.
Let $e$ be an hyperedge of $G$ of size less than $R$.
We arbitrarily introduce a new hyperedge of size $R$ that contains all the vertices of $e$ and has value $w(e)$.
When having only hyperedges of rank $R$, we just take their average value and assign it to each possible subset of rank $R$ as its new value.
The proof of Lemma~\ref{l:canonicalPLE_Symmetric} follows by symmetry of $f$ and by uniqueness of $g$ after being modified.
\end{proof}

\begin{lemma} \label{l:negExampleNminus1}
Let $r, R \in \mathbb{N}$.
If there exists a monotone symmetric set function $f: \{0, \ldots, m\} \to \nonnegR$ of rank $r$ with no positive lower envelope of rank $R$,
then there exists a monotone symmetric set function $f': \{0, \ldots m'\} \to \nonnegR$ of rank $r$, such that any positive hypergraph set function of rank $R$
that obeys $g(m')=f'(m')$, must have $g(m'-1) > f'(m'-1)$.
\end{lemma}
\begin{proof}
Let $f: \{0, \ldots, m\} \to \nonnegR$ be a monotone symmetric set function of rank $r$ with no positive lower envelope of rank $R$, such that $m$ is the smallest possible.
Let $g$ be the canonical positive lower envelope of rank $R$ of $f$. Then, there exists $R \le k \le m-1$ for which $g(k) > f(k)$. 
We claim that since $m$ is the smallest possible, we can just take $k = m-1$ for $f$ itself. Suppose otherwise that $g(m-1) \le f(m-1)$, and for some other $k < m-1$ it holds that $g(k) > f(k)$. Scale $g$ by multiplying it by $f(m-1)/g(m-1)$ and denote it by $g'$. Then, it holds that $g'(m-1) = f(m-1)$ and $g'(k) > f(k)$. Additionally, $g'$ is the canonical positive lower envelope of rank $R$ of $f$ defined on $\{0, \ldots, m-1\}$. Therefore, by Lemma~\ref{l:canonicalPLE_Symmetric}, $f$ defined on $\{0, \ldots, m-1\}$ is also a function of rank $r$ with no positive lower envelope of rank $R$. 
Contradiction.
The proof of Lemma~\ref{l:negExampleNminus1} follows.
\end{proof}

\begin{corollary} \label{c:PleSufficient}
Let $r \geq 1$.
In order to show that any monotone symmetric set function of rank $r$ has a positive lower envelope of rank $R$,
it is sufficient to show that any monotone symmetric set function $f: \{0, \ldots, m\} \to \nonnegR$ of rank $r$ obeys
$f(m-1) \geq \frac{m-R}{m}f(m)$.
\end{corollary}

\begin{proof}
Let $f: \{0, \ldots, m\} \to \nonnegR$ be a monotone symmetric set function of rank $r$.
Let $g$ be the canonical positive lower envelope of rank $R$ of $f$.
Observe that $g(m-1) = \binom{m-1}{R} \frac{f(m)}{\binom{m}{R}} = \frac{m-R}{m}f(m)$.
The proof of Corollary~\ref{c:PleSufficient} follows by Lemma~\ref{l:negExampleNminus1}.
\end{proof}

From Corollary~\ref{c:PleSufficient} together with Lemma~\ref{l:canonicalPLE_Symmetric}, we have: 
\begin{corollary} \label{c:SymmetricPleClassificationByN-1}
Let $r \geq 1$.
Any monotone symmetric set function of rank $r$ has a positive lower envelope of rank $R$,
if and only if any monotone symmetric set function $f: \{0, \ldots, m\} \to \nonnegR$ of rank $r$ obeys
$f(m-1) \geq \frac{m-R}{m}f(m)$.
Moreover, if a \emph{specific} function $f: \{0, \ldots, m\} \to \nonnegR$ does not obey $f(m-1) \geq \frac{m-R}{m}f(m)$, it has no positive lower envelope of rank $R$.
\end{corollary}

\subsection{Upper bound for symmetric hypergraph functions with rank $k$}
\label{sec:symmetric-upper}

In this section we prove Theorem~\ref{t:PleForAnyMonotoneSymmetricFunction}, which asserts that every monotone symmetric set function of rank $r$ (positive and negative) is $\MOPH$-$(3r^2)$.

\begin{proofof} {Theorem~\ref{t:PleForAnyMonotoneSymmetricFunction}}
The proof is by the positive lower envelope technique (i.e. Proposition~\ref{lem:lower-envelope}).
Let $f$ be a monotone symmetric set function of rank $r$, and let $H$ be its hypergraph representation. 
Consider the following monotone symmetric set function $g$ defined by its positive hypergraph representation $G=(V_G,E_G,w_G)$ that has hyperedges of rank $R=3r^2$ of value $f(U) / \binom{m}{R}$ and no other (non-zero) hyperedges. 

We claim that $g$ is a lower envelope for $f$.\footnote{Proof for any restriction of $f$ to a subset is essentially the same.} 
There are three conditions to check. Two of them trivially hold, namely, $g(0) = f(0) = 0$, and $g(m) = \binom{m}{R} w_G(R) = f(m)$. 
The remaining condition requires that $g(k) \le f(k)$ for every $1 \le k \le m-1$. 
This trivially holds for $k < R$ because in this case $g(k) = 0$, whereas $f(k) \ge 0$. 
Hence the main content of our proof is to establish the inequality $g(k) \le f(k)$ for every $R \le k \le m-1$.

Suppose for the sake of contradiction that there is some $f$ that serves as a negative example, namely, that for this $f$ there is $R \le k \le m-1$ for which $g(m) > f(m)$. 
By Corollary~\ref{c:PleSufficient} it is sufficient to show that $f(m-1) \geq \frac{m-R}{m} f(m)$, in order to derive a contradiction.
We now develop some machinery for this aim. 

\begin{prop}
\label{p:polynomial}
There is some polynomial $F$ of degree at most $r$ such that $F(X) = f(x)$ whenever $x \in \{0,1, \ldots, m\}$.
\end{prop}

\begin{proof}
The polynomial is $F(x) = \sum_{i=1}^r \binom{x}{i}h(i)$, where $h(i)$ are constants derived from the hypergraph representation $H$, 
and $\binom{x}{i}$ is the polynomial $\frac{1}{i!}x(x-1) \ldots (x - i + 1)$. This concludes the proof of Proposition~\ref{p:polynomial}.
\end{proof}

Given that $F(x)$ is a polynomial of degree at most $r$, we shall use Markov's inequality regarding derivatives of polynomials~\cite{markov}.

\begin{theorem} [Markov~\cite{markov}] \label{t:markov}
Let $p$ be a polynomial of degree $d$ and let $p'$ denote its derivative. Then
$$\max_{-1 \le x \le 1} |p'(x)| \le d^2\max_{-1 \le x \le 1} |p(x)|$$
\end{theorem}

By a simple transformation of the range on which the polynomial is defined, Theorem~\ref{t:markov} implies that for our polynomial $F(x)$ the following holds:

\begin{equation}\label{eq:markov}
\max_{0 \le x \le m} |F'(x)| \le \frac{2r^2}{m}\max_{0 \le x \le m} |F(x)|
\end{equation}

Let $MAX = \max_{0 \le x \le m} |F(x)|$ and let $0 \le y \le m$ be such that $|F(y)| = MAX$. If $y$ is an integer then monotonicity of $f$ (and hence of $F$ on integer points) implies that $MAX = f(m)$. 
However, $y$ need not be integer. In that case $i < y < i+1$ for some $0 \le i \le m-1$. Let $max = \max \{ |F(i)|,|F(i+1)| \}$. 
Then,
$$MAX \le max + \frac{1}{2}\max_{i \le x \le i+1}|F'(x)| \le f(m) + \frac{1}{2}\max_{0 \le x \le m}|F'(x)| \le f(m) + \frac{r^2}{m} MAX$$
where the last inequality was derived from inequality~(\ref{eq:markov}). As $m \ge R = 3r^2$, we obtain that $MAX \le 3f(m)/2$.

Now we use inequality~(\ref{eq:markov}) to bound $f(m-1) = F(m-1)$ from below.

$$F(m-1) \ge f(m) - \max_{0 \le x \le m}F'(x) \ge f(m) - \frac{2r^2}{m} MAX \ge f(m) - \frac{3r^2}{m}f(m)$$

Using $R = 3r^2$ we have that $f(m - 1) \ge (1 - \frac{R}{m})f(m) = g(m-1)$, as desired.
This concludes the proof of Theorem~\ref{t:PleForAnyMonotoneSymmetricFunction}.
\end{proofof}

\section{Proof of PoA Lower Bound Theorem \ref{thm:poa-lb}}
\label{sec:poa-lb}
\begin{proof}
Consider a projective plane of order $k$.
It has $k(k-1) + 1$ items, $k(k-1) + 1$ bundles, each bundle contains $k$ items, each item is contained in $k$ bundles, every two bundles intersect in exactly one item (and for every two items there is exactly one bundle that contains them, a fact not needed in the proof).

Suppose there are $k(k-1) + 1$ players, each desiring a distinct bundle and valuing it at~1.
We argue that the following is a mixed Nash equilibrium for simultaneous first price auction.
Each player selects at random a value $x \in [0,\frac{1}{k}]$ distributed as $Pr[x \le t] = (kt)^{\frac{1}{(k-1)^2}}$, and bids this value $x$ on each of the items in his bundle.
Fix a player $i$. We show that, given that other players follow the Nash strategy, every value of $x$ gives player $i$ expected utility exactly~0.
For every item in player $i$'s bundle, player $i$ competes with $k-1$ additional players who play according to the Nash strategy, and therefore wins this item with probability $(kx)^{\frac{1}{(k-1)}}$.
Thus, player $i$ wins his entire desired bundle (of $k$ items) with probability $(kx)^{\frac{k}{(k-1)}}$; this is exactly his expected value.
The expected number of items player $i$ wins is $k \cdot (kx)^{\frac{1}{(k-1)}}$, amounting to an expected payment of $(kx)^{\frac{k}{(k-1)}}$.
It follows that the expected utility --- the probability of getting the bundle minus the expected payments --- is~0, for every value $x \in [0,\frac{1}{k}]$, as claimed.

It remains to show that whenever a player does not bid exactly the same value on all his items, the expected utility is negative.
To see this, fix a player $i$, two arbitrary items in his bundle, and his bids on the other $(k-2)$ items in his bundle (not necessarily equal bids).
Suppose that all other players bid according to the Nash strategy, and let $x$ and $y$ denote player $i$'s bids on the two designated items.
It suffices to show that $x=y$ in every best response of player $i$.
By the best response condition, the derivative of player $i$'s expected utility with respect to $x$ equals~0, and the same holds with respect to $y$.
The expected utility can be expressed as follows:
$$
\alpha x^{\frac{1}{k-1}} y^{\frac{1}{k-1}} - \beta - x(kx)^{\frac{1}{k-1}} - y(ky)^{\frac{1}{k-1}},
$$
where $\alpha$ and $\beta$ are constants that do not depend on $x$ and $y$.
The derivative of the expected utility with respect to $x$ is
$$
\frac{\alpha}{k-1} y^{\frac{1}{k-1}} x^{\frac{2-k}{k-1}} - \frac{k}{k-1} (kx)^{\frac{1}{k-1}}.
$$
The derivative with respect to $y$ is obtained by swapping $x$ and $y$ in the last expression.
Equating both derivatives to~0 and solving the obtained system of equations gives us $x=y$ for every $k\neq 0$, as desired.


Given the above, consider $k$ such independent projective planes (hence there are $k^2(k-1) + k$ players), each with such a Nash equilibrium. Now add $k(k-1) + 1$ auxiliary players, where auxiliary player $i$ wants a bundle composed of the $i$th item of each projective plane. The optimal solution is to give each auxiliary player his desired bundle, giving value $k(k-1) + 1$. However, given the equilibrium of the original players, the unique optimal strategy for the auxiliary players is not to bid at all.
Indeed, they are faced with strictly more competition than the Nash players, and the Nash players have expected utility of~0.
(For concreteness, following the same reasoning as above, a player's best response is to bid equally on all items in his bundle.
For any bid $x$, he derives an expected value of $(kx)^{\frac{k^2}{(k-1)^2}}$, which is strictly smaller than his expected payment, being $(kx)^{\frac{k}{(k-1)^2}+1}$.)
Hence in the Nash solution in each projective plane only one player gets value, giving a total value of $k$, and the price of anarchy for this example is $\frac{k(k-1) + 1}{k} = k - 1 + \frac{1}{k}$.
\end{proof}

\section{Smoothness of the Simultaneous Auction}
\label{sec:smoothness}\label{SEC:SMOOTHNESS}

We now prove Theorem~\ref{t:poa} in its full generality by showing that the simultaneous first price auction is actually a Smooth Mechanism as defined in \cite{Syrgkanis2013}. This implies
an efficiency guarantee that extends to Bayes-Nash equilibria as well as to no-regret learning outcomes. For completeness we first present the definition and the main implication of smooth mechanisms.

\begin{defn}[\cite{Syrgkanis2013}]
A mechanism $M$ is $(\lambda,\mu)$-smooth if for any valuation profile $v$ there exists an action profile $a_i^*(v)$ such that for all $a\in \A$:
\begin{equation}
\sum_i u_i(a_i^*(v),a_{-i};v_i)\geq \lambda \opt(v) - \mu \sum_i P_i(a)
\end{equation}
\end{defn}

\begin{theorem}[\cite{Syrgkanis2013}] If a mechanism is $(\lambda,\mu)$-smooth then the Bayes-Nash and the correlated
price of anarchy is at most $\mu/\lambda$.
\end{theorem}

We will analyze a generic simultaneous single-item auction where each item $j$ is sold via an auction with some allocation and payment rule (i.e., not necessarily the first-price auction).
We will show that if each individual single-item auction is a $(\lambda,\mu)$-smooth mechanism
and bidders have $\MOPH$-$k$ valuations then the simultaneous single-item auction
is $\left(1-k+\min\left\{\lambda,1\right\}\cdot k,\mu\right)$-smooth.
We can then invoke known results about the smoothness of the first-price auction to complete the proof of Theorem~\ref{t:poa}.

%

\begin{theorem}\label{thm:smooth-poa}
The game defined by running $m$ simultaneous $(\lambda,\mu)$-smooth single-item auctions is $\left(1-k+\min\left\{\lambda,1\right\}\cdot k,\mu\right)$-smooth, when players have $\MOPH$-$k$ valuations.
\end{theorem}
\begin{proof}
We will first show that if we prove smoothness of the simultaneous auction for hypergraph-$k$ valuations then this immediately
implies smoothness for $\MOPH$-$k$ valuations (Lemma \ref{lem:hypergraph-to-MOPH}).  We will then complete the proof of Theorem \ref{thm:smooth-poa} by proving smoothness for the class of hypergraph-$k$ valuations.
\begin{lemma}\label{lem:hypergraph-to-MOPH}
If a simultaneous single-item auction game is $(\lambda,\mu)$-smooth for $\PH$-$k$ valuations then it is also
$(\lambda,\mu)$-smooth for the class of $\MOPH$-$k$ valuations.
\end{lemma}
\begin{proof}
Consider a valuation profile $\vals$.  Let $X^*(\vals)$ be the optimal allocation for valuation profile $\vals$, so that $X_i^*(\vals)$ is the allocation to player $i$ in this optimal allocation.  Since $\val_i$ is $\MOPH$-$k$, we have $\val_i(X_i^*(\vals)) = \max_{\ell \in \Ell} \val_i^{\ell}(X_i^*(\vals))$.  Choose $\val_i^* \in \arg\max_{\ell\in\Ell} \val_i^{\ell}(X_i^*(\vals))$, so $\val_i^*$ is the additive valuation supporting player $i$'s value for set $X_i^*(\vals)$.
Then, by definition, $\val_i(X_i^*(\vals)) = \val_i^*(X_i^*(\vals))$ and for any set $T\subseteq [m]$, $\val_i(T)\geq  \val_i^*(T)$.

Observe that a player's utility under valuation $v_i$ is at least as much as his utility under $v_i^*$: $u_i(a;v_i)\geq u_i(a;v_i^*)$.
Now consider the smoothness deviations $a_i^*(\vals^*)$ that correspond to valuation profile $\vals^*=(v_1^*,\ldots,v_n^*)$, which exist by the assumption that the mechanism is smooth for hypergraph-$k$ valuations. Then we get that for any action profile:
\begin{align*}
\sum_i u_i(a_i^*(\vals^*),a_{-i};v_i)\geq~& \sum_i u_i(a_i^*(\vals^*),a_{-i};v_i^*)\geq \lambda\opt(\vals^*)-\mu\sum_i P_i(a)\\
\geq~& \lambda SW(X^*(\vals);\vals^*)-\mu \sum_i P_i(a) = \lambda \opt(\vals) - \mu \sum_i P_i(a).
\end{align*}
\end{proof}

We will now move on to proving smoothness for the class of $\PH$-$k$ valuations.
Consider a $\PH$-$k$ valuation profile $\vals$ and for each valuation $v_i:2^M\to \nonnegR$, let $w_i:2^M\to \nonnegR$ be its positive hypergraph-$k$ representation. Also let $\opti$ be the optimal set of items for each player $i$. Consider an action profile $a=(a^j)_{j\in [m]}$ on each auction $j$ and each player deviating to some strategy $\tilde{a}_i=\left(\tilde{a}_i^j\right)_{j\in [m]}$. Denote with $\Pr(S,a)$ the probability of winning set $S$ under a randomized action profile $a$. Also we denote with $\ESk{S}$ the collection of subsets of a set $S$ of size at most $k$. Then by analyzing the utility of the player and applying the union bound in a generic way, we get
\begin{align*}
u_i(\tilde{a}_i,a_{-i})=~& \sum_{S\in \ESk{M}}w_i(S)\cdot \Pr(S,(\tilde{a}_i,a_{-i}))- \sum_{j\in M}P_i^j(\tilde{a}_i^j,a_{-i}^j)\\
\geq~& \sum_{S\in \ESk{\opti}}w_i(S)\cdot \Pr(S,(\tilde{a}_i,a_{-i}))- \sum_{j\in M}P_i^j(\tilde{a}_i^j,a_{-i}^j)\\
\geq~& \sum_{S\in \ESk{\opti}}w_i(S)\cdot \left(1-\sum_{j\in S}\left(1-\Pr\left(\{j\},(\tilde{a}_i,a_{-i}\right)\right)\right)- \sum_{j\in M}P_i^j(\tilde{a}_i^j,a_{-i}^j)\\
=~& \sum_{S\in \ESk{\opti}}w_i(S)\cdot(1-|S|)+\sum_{S\in \ESk{\opti}}w_i(S)\sum_{j\in S}\Pr\left(\{j\},(\tilde{a}_i,a_{-i})\right)- \sum_{j\in M}P_i^j(\tilde{a}_i^j,a_{-i}^j)\\
=~& \sum_{S\in \ESk{\opti}}w_i(S)\cdot(1-|S|)+\sum_{j\in M}\left\{\left(\sum_{S\in \ESk{\opti}: S\ni j}w_i(S)\cdot\mathrm{1}_{j\in \opti}\right)\cdot\Pr\left(\{j\},(\tilde{a}_i^j,a_{-i}^j)\right)-P_i^j(\tilde{a}_i^j,a_{-i}^j)\right\}.
\end{align*}

Consider the second summand in the above expression.  For each $j$, this corresponds to the utility of player $i$ under a deviation to action $\tilde{a}_i^j$, of a single-item auction, in which only
player $i$ has a value of $\sum_{s\in \ESk{\opti}: s\ni j}w_s$ for the auction and everyone else has a value of $0$.  Summing up over all players, this is a sum of deviating utilities for each player, over multiple single-item auctions:
\begin{align*}
\sum_{j\in M}\sum_{i\in P} \left\{\left(\sum_{S\in \ESk{\opti}: S\ni j}w_i(S)\cdot\mathrm{1}_{j\in \opti}\right)\cdot\Pr\left(\{j\},(\tilde{a}_i^j,a_{-i}^j)\right)-P_i^j(\tilde{a}_i^j,a_{-i}^j)\right\}.
\end{align*}
Thus, if we set $\tilde{a}^j$ equal to the smoothness deviation for the above valuation profile, we get that the latter expression is at least
\begin{align*}
\sum_{j\in M}\left\{\lambda\cdot \sum_{i\in P} \left(\sum_{S\in \ESk{\opti}: S\ni j}w_i(S)\cdot\mathrm{1}_{j\in \opti}\right) - \mu \sum_{i\in P} P_i^j(a^j)\right\}=
\lambda \sum_{i\in P}\sum_{S\in \ESk{\opti}}w_i(S)\cdot |S| - \mu \sum_{i\in P} P_i(a).
\end{align*}
Combining the above we get:
\begin{align*}
\sum_{i\in P} u_i(\tilde{a}_i,a_{-i}) \geq~& \sum_{i\in P}\sum_{S\in \ESk{\opti}}w_i(S)\cdot(1-|S|)+\lambda \sum_{i\in P} \sum_{S\in \ESk{\opti}}w_i(S)\cdot |S| - \mu \sum_{i\in P} P_i(a)\\
=~& \sum_{i\in P}\sum_{S\in \ESk{\opti}}w_i(S)\cdot\left(1-(1-\lambda)|S|\right) - \mu \sum_{i\in P} P_i(a).
\end{align*}
If $\lambda<1$ then we use the fact that $|S|\leq k$ to get the $(1-k+\lambda k,\mu)$-smoothness property, otherwise we can simply ignore the term $(1-\lambda)|S|$ and get the $(1,\mu)$-smoothness property, completing the proof of Theorem \ref{thm:smooth-poa}.
\end{proof}

We now show how to use Theorem \ref{thm:smooth-poa} to prove Theorem~\ref{t:poa}.
For the case in which each single-item auction is a first price auction then we know by \cite{Syrgkanis2013} that each auction is $(\beta\cdot(1-e^{-1/\beta}),\beta)$-smooth for any $\beta$.
Thus we get that for any $\beta$ the simultaneous first price auction with $\MOPH$-$k$ valuations is $(1-(1-\beta\cdot(1-e^{-1/\beta}))\cdot k,\beta)$-smooth. 
Substituting $\beta=\log(\frac{k}{k-1})$, we get a bound of $\frac{1}{1-(k-1)\log(\frac{k}{k-1})}\leq k(2-e^{-k})$ on the price of anarchy, as desired.
This establishes the proof of Theorem~\ref{t:poa}.

\section{Composition of General Mechanisms}\label{sec:composition}\label{SEC:COMPOSITION}

Our analysis can be extended beyond simultaneous single-item auctions to the simultaneous composition of general mechanisms, such as position auctions.
We consider the mechanism defined by running $m$ different mechanisms simultaneously. Each mechanism $M^j$ has
its own feasible set of allocations $\X^j\subseteq \X_1^j\times\ldots \times \X_n^j$,  action spaces $\A^j$, allocation function $X^j:\A^j\rightarrow \X^j$ and payment
function $P^j:\A^j\rightarrow \R^n_+$. Each player $i$ has a valuation over allocations in different mechanisms, given by $v_i:\X_i^1\times\ldots\times \X_i^m\rightarrow \R_+$. We consider the natural generalization of $\MOPH$-$k$ valuations:
\begin{defn} A valuation is $\MOPH$-$k$ across mechanisms if for any $x_i\in \X_i^1\times\ldots\times \X_i^m$
\begin{equation}
v_i(x_i^1,\ldots,x_i^m) = \max_{\ell \in \Ell_i} \textstyle{\sum_{e\in E_{\ell}} v_{i}^{e,\ell}(x_i^e)}
\end{equation}
where $\Ell_i$ is some arbitrary index set, $E_{\ell}\subseteq \{S\subseteq M: |S|\leq k\}$, $x_i^e = (x_i^j)_{j\in e}$ is the vector of outcomes on the mechanisms in the set $e$ and for all $e\in E_{\ell}$, $v_{i}^{e,\ell}(x_i^e)\geq 0$.
\end{defn}

We show that if each allocation space $\X_i^j$ is partially ordered and the value functions $v_{i}^{e,\ell}(x_i^e)$ are monotone coordinate-wise with respect to this partial order, then if each mechanism is $(\lambda,\mu)$-smooth for the class of monotone valuations then this implies that the simultaneous composition is  $\left(1-k+\min\left\{\lambda,1\right\}\cdot k,\mu\right)$-smooth.
\begin{theorem}\label{thm:composition}
Consider the simultaneous composition of $m$ mechanisms each being $(\lambda,\mu)$-smooth for any monotone valuation
with respect to some partial order of the allocation space. If players have $\MOPH$-$k$ valuations across mechanisms such that
$v_i^{e,\ell}(\cdot)$ are monotone coordinate-wise with respect to each partial order, then the composition is $\left(1-k+\min\left\{\lambda,1\right\}\cdot k,\mu\right)$-smooth
\end{theorem}
\begin{proof}
For simplicity we will consider a $\PH$-$k$ valuation profile $\vals$ across mechanisms, i.e. for each player $i$ we have:
\begin{equation}
v_i(x) = \sum_{e\in E_i}v_i^e(x_i^e),
\end{equation}
where $E_i\subseteq \{S\subseteq M: |S|\leq k\}$ and $x_i^e=(x_i^j)_{j\in e}$.

Let $\tilde{x}_i=(\tilde{x}_i^j)_{j\in [m]}$ be the optimal allocation of each player $i$. Consider an action profile $a=(a^j)_{j\in M}$ on each auction $j$ and each player deviating to some strategy $\tilde{a}_i=\left(\tilde{a}_i^j\right)_{j\in [m]}$. Then following an analysis similar to the proof of Theorem \ref{thm:smooth-poa} we
can obtain the following lower bound a player's utility from the deviation:
\begin{align*}
u_i(\tilde{a}_i,a_{-i})=~& \sum_{e\in E_i}\sum_{x_i^e}v_i^e(x_i^e)\cdot \Pr\left(X_i^e(\tilde{a}_i,a_{-i})=x_i^e\right)- \sum_{j\in M}P_i^j(\tilde{a}_i^j,a_{-i}^j)\\
\geq~& \sum_{e\in E_i}v_i^e(\tilde{x}_i^e)\cdot \Pr\left(X_i^e(\tilde{a}_i,a_{-i})\succeq \tilde{x}_i^e\right)- \sum_{j\in M}P_i^j(\tilde{a}_i^j,a_{-i}^j)\\
\geq~& \sum_{e\in E_i}v_i^e(\tilde{x}_i^e)\cdot \left(1-\sum_{j\in e}\left(1- \Pr\left(X_i^j(\tilde{a}_i^j,a_{-i}^j)\succeq \tilde{x}_i^j\right)\right)\right)- \sum_{j\in M}P_i^j(\tilde{a}_i^j,a_{-i}^j)\\
=~& \sum_{e\in E_i}v_i^e(\tilde{x}_i^e)\cdot(1-|e|)+\sum_{e\in E_i}v_i^e(\tilde{x}_i^e)\sum_{j\in e}\Pr\left(X_i^j(\tilde{a}_i^j,a_{-i}^j)\succeq \tilde{x}_i^j\right)- \sum_{j\in M}P_i^j(\tilde{a}_i^j,a_{-i}^j)\\
=~&  \sum_{e\in E_i}v_i^e(\tilde{x}_i^e)\cdot(1-|e|)+\sum_{j\in M}\left\{\left(\sum_{e\ni j}v_i^e(\tilde{x}_i^e)\right)\cdot\Pr\left(X_i^j(\tilde{a}_i^j,a_{-i}^j)\succeq \tilde{x}_i^j\right)-P_i^j(\tilde{a}_i^j,a_{-i}^j)\right\}.
\end{align*}

Summing up over all players we observe that the second summand in the above expression will correspond to deviating utilities of individual single-item auctions, where each player unilaterally deviates to  $\tilde{a}_i^j$ and in which every player has a valuation of $\sum_{e\ni j}v_i^e(\tilde{x}_i^e)$ for getting any allocation $x_i^j\succeq \tilde{x}_i^j$ and $0$ otherwise. The latter is a monotone valuation and hence we can set the local deviating utilities at each mechanism to the smoothness deviations for the latter monotone valuation profiles and get:
\begin{align*}
\sum_{i} u_i(\tilde{a}_i,a_{-i}) \geq~& \sum_i \sum_{e\in E_i}v_i^e(\tilde{x}_i^e)(1-|e|)+\lambda \sum_i  \sum_{e\in E_i}v_i^e(\tilde{x}_i^e)\cdot |e| - \mu \sum_i P_i(a)\\
=~& \sum_i\sum_{e\in E_i}v_i^e(\tilde{x}_i^e)\cdot\left(1-(1-\lambda)|e|\right) - \mu \sum_i P_i(a).
\end{align*}
The remainder of the proof precisely follows the proof of Theorem \ref{thm:smooth-poa}.
\end{proof}

\paragraph{Example: position auctions.}
An example of such a composition of mechanisms is the composition of $m$ position auctions. Suppose that each position auction
is a first price pay-per-impression auction. Syrgkanis and Tardos \cite{Syrgkanis2013} showed that such a mechanism
is $(1/2,1)$-smooth. We extend this analysis to show that it is $(1-\frac{1}{2\beta},\beta)$-smooth for any $\beta\geq 1$.
\begin{lemma}\label{lem:position-auction}
The first-price pay-per-impression position auction is a $(1-\frac{1}{2\beta},\beta)$-smooth mechanism for any $\beta\geq 1$.
\end{lemma}
\begin{proof}
Consider a bid profile $b$ and let $j_i^*$ be the optimal position of player $i$ and let $\pi(j)$ be the player that
gets slot $j$ under bid profile $b$. Suppose that each player deviates to bidding a random bid $b_i'$, uniformly in $[0,\frac{v_{ij_i^*}}{\beta}]$
and let $f(t)$ denote the density function of the random bid. If the random bid $t$ of a player is $b_{\pi(j_i^*)}<t$ then
player $i$ wins his optimal slot or a higher slot and hence his value is at least $v_{ij_i^*}$ by monotonicity of the valuation.

Thus a player's utility from this deviation is at least:
\begin{align*}
u_i(b_i',b_{-i})\geq \int_{b_{\pi(j_i^*)}}^{\frac{v_{ij_i^*}}{\beta}}v_{ij_i^*}f(t)dt - \frac{v_{ij_i^*}}{2\beta}=
 \int_{b_{\pi(j_i^*)}}^{\frac{v_{ij_i^*}}{\beta}}\beta\cdot dt - \frac{v_{ij_i^*}}{2\beta} = \left(1-\frac{1}{2\beta}\right)v_{ij_i^*}-\beta\cdot b_{\pi(j_i^*)}
\end{align*}
Summing over all players we get the $(1-\frac{1}{2\beta},\beta)$-smoothness property.
\end{proof}
Combined with Theorem \ref{thm:composition}, we get that for any $\MOPH$-$k$ valuation across position auctions the simultaneous position auction mechanism is $(1-\frac{k}{2\beta},\beta)$-smooth, yielding a price of anarchy bound of $2k$ for $\beta=k$.

\section{Non-Monotone Valuations}\label{sec:app-non-monotone}\label{sec:non-monotone}


As we saw in Section \ref{sec:proofs}, positive lower envelopes are a useful tool for showing that a monotone function admits a $\MOPH$-$k$ representation. Observe that while $\MOPH$-$k$ contains only monotone functions, even a non-monotone function can admit a positive lower envelope for any restriction to a subset. 
This observation motivates the following hierarchy of (not necessarily monotone) set functions.

\begin{defn} [Positive Lower Envelope $k$ ($\PLE$-$k$) class]
Let $f: 2^M \to \nonnegR$ be a set function.
$f$ is in $\PLE$-$k$ if any restriction of it to a subset of its ground set admits a positive lower envelope of rank $k$.
\end{defn}

While monotonicity is a very natural assumption in the context of combinatorial auctions, 
hierarchies of non-monotone functions can have potential applications in more general contexts
such as set function minimization and maximization.

In Propositions~\ref{p:PLEtoMOH} and~\ref{p:MOHtoPLE} we showed that $\PLE$-$k$ is equivalent to $\MOPH$-$k$ for monotone functions.
Thus $\MOPH$-$k$ is exactly the monotone part of the $\PLE$-$k$ class. Moreover, we point out that both our algorithmic approximation result (Theorem~\ref{t:algorithmic}) and our price of anarchy result (Theorem~\ref{t:poa}) apply more generally to the class of $\PLE$-$k$ functions, and in particular do not require monotonicity (assuming that players are able to drop out of an auction and receive no items).

With respect to expressiveness, exploring the expressive power of the $\PLE$ hierarchy is an interesting space for future research. In particular, one might explore the connection between $\PLE$-$k$ and the class of functions with a given positive and/or negative hypergraph rank.  
To this end, we show some results on this connection.
We first observe that a strong positive result similar to Theorem \ref{t:PleForAnyMonotoneSymmetricFunction},
cannot possibly hold for non-monotone non-symmetric set functions.
That is, for any constant $k$, there exist a non-negative set function of rank $2$ that is not $\PLE$-$k$.

This is not the case for non-negative functions with positive rank $1$ and arbitrary negative rank, which can be shown to be in $\PLE$-$1$.

\begin{theorem} \label{t:negResult_NonNeg}
Let $k \in \mathbb{N}$.
There exists a non-negative set function $f_k:2^M \to \nonnegR$ of rank~2 that has no positive lower envelope of rank $k$.
\end{theorem}

\begin{proof} 
Given $k$, we have the following hypergraph representing $f_k$.
We set some $m \geq k + 3$ to be the number of items in the ground set $M$ of $f$.
We set a special item $j \in M$. Let $v_j$ be its corresponding vertex in the hypergraph representation.
We set the value of the singleton hyperedge $v_j$ to be $\binom{k+1}{2}$ and the value of any other singleton hyperedge to be~0.
We set the value of any rank~2 edge to be~1 if it does not contain $v_j$ and $-k$ otherwise.
There are no other hyperegdes, so the set function represented by this representation is indeed of rank~2.
We show that $f_k$ is non-negative.
It is trivial that the value of any subset of items not containing $j$ is non-negative.
Let $S$ be a subset of items such that $j \in S$ and let $l = |S|-1$.
Then, $f_k(S) = \binom{l}{2} - l k + \binom{k+1}{2}$.
Clearly, for $l=k$ and for $l=k+1$, it equals~0, and for any other integer it is strictly greater than~0.
Therefore, $f_k$ is indeed non-negative and it has value~0 for any subset containing $j$ and exactly other $k$ (or $k+1$) distinct items.
Assume towards a contradiction that $f_k$ has a positive lower envelope of rank $k$. Let $g_k$ be such a positive lower envelope.
Then, $g_k$ must be monotone, since its hypergraph representation contains only hyperedges of non-negative value.
Therefore, it has value~0 for any subset of up to size $k+1$, since any such subset is either of value~0 by $f_k$ (in case it contains $j$) or contained in a subset of value~0 by $f_k$ (in case it does not contain $j$).
Therefore and since the rank of $g_k$ is assumed to be at most $k$, any hyperedge in its hypergraph representation must be of value~0.
But, since $m \geq k + 3$, we have $f_k(M) > 0$. contradiction.
Therefore, $f_k$ does not have any positive lower envelope of rank up to $k$, as desired.
\end{proof}

Moreover, the proof of Theorem~\ref{t:laminar_nonneg} extends to non-negative functions in $\PLE$-$k$. That is, any non-negative function with positive rank $k$ and arbitrary negative rank is in $\PLE$-$k$, if negative hyperedges are laminar.

\begin{theorem} \label{t:rank1nonneg}
There is a positive lower envelope of rank~1 for any {\em non-negative} set function of positive rank~1.
\end{theorem}

\begin{proof}
Let $f:2^M \to \nonnegR$ be a non-negative set function of positive rank~1 and let ${\cal E}$ be the largest value dividing all images of $f$.
We refer to each value ${\cal E}$ as ``unit of value''.
Let $G_f$ be the hypergraph representation of $f$.
We show how to discard all the negative hyperedges of $G_f$.
We build the following bipartite graph $H(V, W, E)$.
For each unit of value of every vertex of $G_f$, we have a vertex in $V$.
For each unit of value of every negative hyperedge, we have a vertex in $W$.
For simplicity, we refer sometimes to vertices in $H$ as units of value.
Let $v \in V$ and $w \in W$. We have an edge $(v,w)$ if and only if the hyperedge in $G_f$ that corresponds to the unit of value of $w$ contains the vertex that corresponds to the unit of value of $v$.
We show that there exists a matching in $H$ that saturates $W$.
Such a matching will enable us to remove all the negative hyperedges of $G_f$ by reducing also the value of vertices of $G_f$ that their corresponding vertices participate in the matching in $H$.
Assume towards a contradiction that there is no matching in $H$ that saturates $W$.
Then, by Hall's theorem, there exists $W' \subseteq W$ such that $|W'| > |N_H(W')|$, where $N_H(W')$ are the neighbours of $W'$ in $H$.
Note that by definition of $H$, $N_H(W')$ contains either zero or all units of value of any single vertex of $G_f$.
Therefore, we can add to $W'$ all units of value of any hyperedge of $H$ that at least one unit of value of it is already containted in $W'$, without increasing $|N_H(W')|$.
Let $W''$ be the resulting subset. Clearly, $|W''| > |N_H(W'')|$.
Moreover $W''$ contains for any vertex of $G_f$ either zero or all units of values of it.
Additionally, as before, $N_H(W'')$ contains for any hyperedge of $G_f$ either zero or all units of values of it.
Therefore, for the set $S \subseteq M$ that contains exactly the items that corresponds to the vertices of $G_f$ contained in hyperedges corresponding to the vertices of $W''$,
it must be that $f(S) < 0$.\footnote{Note that all the positive hyperedges induced by $S$ are indeed represented appropriately in the matching, since we have positive rank~1. If we had positive rank strictly greater than~1, it might have been that a positive hyperedge will be induced by $S$ and not by the matching. This would have happened if e.g. part of its vertices were contained in one negative hyperedge participating in the matching and all the rest in another, but there was no single negative hyperedge containing all its vertices.} Contradiction.
Therefore, there is a matching saturating $W$ in $H$.
Note that if we discard all the negative hyperedges and also decrease the value of each of the vertices according to the number of units of value of it participating in the matching,
we will never increase the value of any subset of $M$ by $f$.
Moreover, we will not change $f(M)$.
Therefore, we can discard all the negative hyperedges and then return as a hypergraph representation of a positive lower envelope of $f$ all the positive vertices that survived, with the value of them that survived.
This concludes the proof of Theorem~\ref{t:rank1nonneg}.
\end{proof}

\section{Tight bounds for symmetric set functions with ranks~3 and~4}\label{app:smallRanks}

\subsection{Any monotone symmetric set function of rank~3 is $\MOPH$-$4$}

We show the following theorem and then show its tightness.
\begin{theorem} \label{t:upperBoundRank3}
Let $f : 2^M \to \nonnegR$ be a monotone symmetric set function of rank~3.
Then, there exists a positive lower envelope of $f$ of rank~4 for any restriction of it to a subset of $M$.
That is, any monotone symmetric set function of rank~3 is in $\MOPH$-4.
\end{theorem}
\Omit{
In order to prove Theorem~\ref{t:upperBoundRank3}, we use the characterization of Corollary~\ref{c:SymmetricPleClassificationByN-1} and show an LP that finds a function that is the worst possible.
Then, by duality theorem, we bound the feasible solutions of that LP and conclude the theorem.
The proof is postponed to Appendix~\ref{app:smallRanks}.
}
\begin{proof}
For simplicity of presentation, we prove Theorem~\ref{t:upperBoundRank3} only for $f$ itself, and not for any restriction of it to a subset. The proof for any restriction is essentially the same.
Let $m=|M|$.
Assume towards a contradiction that there is a monotone symmetric set function $f$ of rank~3 with no positive lower envelope of rank~4. We scale $f$ to have $f(m)=m$. Then, by Corollary~\ref{c:PleSufficient}, we may assume without loss of generality that
$f(m-1) < m-4$. 

We consider the following linear program.
\begin{LP} \label{LP:min_rank3}
\ \\\textbf{Minimize} $\binom{m-1}{1}x_1 + \binom{m-1}{2}x_2 + \binom{m-1}{3}x_3$ \\
\textbf{Subject to:}
\begin{description}
\item{Monotonicity:} $\forall_{t \in \{ 0,1, \ldots m-1 \}} : x_1 + \binom{t}{1}x_2 + \binom{t}{2}x_3 \geq 0$.
\item{Scaling:} $m x_1 + \binom{m}{2}x_2 + \binom{m}{3}x_3 = m$
\end{description}
\end{LP}

Let $f^*$ be an optimum of Linear Program~\ref{LP:min_rank3}.
Then, it must be that $f^*(m-1) < n-4$.
We consider the dual of Linear Program~\ref{LP:min_rank3}:
\begin{LP} \label{LP:max_rank3}
\ \\\textbf{Maximize} $n \cdot z$ \\
\textbf{Subject to:}
\begin{enumerate}
\item $\sum_{t=0}^{m-1} y_t + \binom{m}{1} z = \binom{m-1}{1}$
\item $\sum_{t=0}^{m-1} \binom{t}{1} y_t + \binom{m}{2} z = \binom{m-1}{2}$
\item $\sum_{t=0}^{m-1} \binom{t}{2} y_t + \binom{m}{3} z = \binom{m-1}{3}$
\item $\forall_t y_t \geq 0$
\end{enumerate}
\end{LP}

Since $f^*(m-1) < m-4$ and from duality theorem, it must be that any feasible solution of Linear Program~\ref{LP:max_rank3} is of value strictly less that $m-4$.
We show a feasible solution to Linear Program~\ref{LP:max_rank3} of value $m-4$ and derive a contradiction.
Specifically, we prove that for any $m > 3$, there is a feasible solution for Linear Program~\ref{LP:max_rank3} with $z=\frac{m-4}{m}$. 
That is, we show that the following equation system has a solution with non-negative variables $y_t$ for 
$t \in \{ 0, \ldots, m-1 \}$, for any $m$:
\begin{equation} \label{eqs:rank3}
\begin{cases}
\sum_{t=0}^{m-1} y_t = 3 \\
\sum_{t=0}^{m-1} t \cdot y_t = m-1 \\
\sum_{t=0}^{m-1} \binom{t}{2} y_t + \binom{m}{3} z = \binom{m-1}{3}
\end{cases}
\end{equation}

We consider the following separate cases with solutions concluded with the aid of a computer\footnote{We used Microsoft .NET together with Gurobi (\cite{gurobi}) and IBM CPLEX and also Wolfram Mathematica.} to solve Linear Program~\ref{LP:max_rank3} for various values of $m$:

\paragraph{{\boldmath $m \mod 3 = 1$}:} 
The solution we consider is the following:
$$
y_t = 
\begin{cases}
\frac{2(m-1)}{m+2} & t = \frac{m-4}{3} \\
1 & t = \frac{m-1}{3} \\
\frac{6}{m+2} & t = 2\frac{m-1}{3} \\
0 & \mbox{otherwise}
\end{cases}
$$

We need to show that:
$$
\begin{cases}
y_{\frac{m-4}{3}} + y_{\frac{m-1}{3}} + y_{2\frac{m-1}{3}} = 3 \\
\frac{m-4}{3} y_{\frac{m-4}{3}} + \frac{m-1}{3} y_{\frac{m-1}{3}} + 2\frac{m-1}{3} y_{2\frac{m-1}{3}} = m-1 \\
\frac{(m-4)(m-7)}{18} y_{\frac{m-4}{3}} + \frac{(m-1)(m-4)}{18} y_{\frac{m-1}{3}} + \frac{(2m-2)(2m-5)}{18} y_{2\frac{m-1}{3}} = \frac{(m-1)(m-2)}{6}
\end{cases}
$$

When substituting the rest of our considered solution, we get the following:
For the first equation we get
$\frac{2m-2 + 6}{m+2} + 1 = 3$, which is obviously true.

For the second one we get
$$\frac{2(m-1)(m-4)}{3(m+2)} + \frac{m-1}{3} + \frac{12(m-1)}{3(m+2)} = m-1 \ . $$
By multiplying both sides by $3(m+2)/(m-1)$ we get
$2(m-4)+m+2+12 = 3m+6$,
which is obviously true.

For the third one we get
$$\frac{2(m-1)(m-4)(m-7)}{18(m+2)} + \frac{(m-1)(m-4)}{18} + \frac{6(2m-2)(2m-5)}{18(m+2)} = \frac{(m-1)(m-2)}{6} \ . $$
By multiplying both sides by $18(m+2)/(m-1)$, we get
$2(m-4)(m-7) + (m-4)(m+2) + 12(2m-5) = 3(m+2)(m-2)$,
which is true by simple calculations.

\paragraph{{\boldmath $n \mod 3 = 2$}:}
The solution we consider is the following:
$$
y_t = 
\begin{cases}
1 & t = \frac{m-5}{3} \\
\frac{2(m-5)}{m-2} & t = \frac{m-2}{3} \\
\frac{6}{m-2} & t = 2\frac{m-2}{3} \\
0 & \mbox{otherwise}
\end{cases}
$$

Correctness follows by substitution at Equation System~\eqref{eqs:rank3} and simple calculations.

\paragraph{{\boldmath $m \mod 3 = 0$}:} 
The solution we consider is the following:
$$
y_t = 
\begin{cases}
\frac{3m-23}{m-3} & t = \frac{m-3}{3} \\
\frac{8}{m} & t = \frac{m-6}{3} \\
\frac{6(m+4)}{m(m-3)} & t = 2\frac{m-3}{3} \\
0 & \mbox{otherwise}
\end{cases}
$$

Correctness follows by substitution at Equation System~\eqref{eqs:rank3} and simple calculations.
This concludes the proof of Theorem~\ref{t:upperBoundRank3}.
\end{proof}

Theorem~\ref{t:upperBoundRank3} is tight by the following theorem.
\begin{theorem} \label{t:tightnessRank3}
There exists a monotone symmetric set function of rank~3 that has no positive lower envelope of rank~3.
\end{theorem}
\begin{proof}
Let $f: \{0, \ldots, m\} \to \nonnegR$ be a set function with the following hypergraph representation.
Hyperedges of rank~1 have value~1;
Hyperedges of rank~2 have value~-1;
Hyperedges of rank~3 have value~1;
That is, $f(x) = x-\binom{x}{2}+\binom{x}{3}$.
It is clear this function is symmetric and that it has rank~3.
We show monotonicity.
From simple calculations, $f(1) = f(2) = f(3) = 1$.
The marginal value for adding the $x^{th}$ item to at least~2 items is:
$$
f(x) - f(x-1) = 1-(x-1)+\binom{x-1}{2}=\binom{x-1}{2}-x+2 = \frac{1}{2}(x-1)(x-2)-x+2=\frac{1}{2}x^2 - 2\frac{1}{2}x + 3 \ .
$$
This is greater than~0 for any $x>3$, as desired.

We show inexistence of positive lower envelope of rank~3.
By Corollary~\ref{c:SymmetricPleClassificationByN-1} it is sufficient to show that $f(m-1) < \frac{m-3}{m}f(m)$.
We show it is true for $n=6$:
$$f(m) = f(6) = 6 - \binom{6}{2} + \binom{6}{3} = 11 \ . $$
$$f(m-1) = f(5) = 5- \binom{5}{2} + \binom{5}{3} = 5 \ . $$
Therefore, 
$$\frac{m-3}{m}f(m) = \frac{1}{2} \cdot 11 = 5.5 > 5 = f(m-1) \ , $$
as desired.
This concludes the proof of Theorem~\ref{t:tightnessRank3}.
\end{proof}

\subsection{Any monotone symmetric set function of rank~4 is $\MOPH$-$6$} 
We show the following theorem and then show its tightness.
\begin{theorem} \label{t:upperBoundRank4}
Let $f : 2^M \to \nonnegR$ be a monotone symmetric set function of rank~4.
Then, there exists a positive lower envelope of $f$ of rank~6 for any restriction of it to a subset of $M$.
That is, any monotone symmetric set function of rank~4 is in $\MOPH$-6.
\end{theorem}

\begin{proof}
For simplicity of presentation, we prove Theorem~\ref{t:upperBoundRank4} only for $f$ itself, and not for any restriction of it to a subset. The proof for any restriction is essentially the same.
Let $m=|M|$.
Assume towards a contradiction that there is a monotone symmetric set function $f$ of rank~4 with no positive lower envelope of rank~6. We scale $f$ to have $f(m)=m$. Then, by Corollary~\ref{c:PleSufficient}, we may assume without loss of generality that
$f(m-1) < m-6$. 

We consider the following linear program (similarly to the case of rank~3).
\begin{LP} \label{LP:min_rank4}
\ \\\textbf{Minimize} $\binom{m-1}{1}x_1 + \binom{m-1}{2}x_2 + \binom{m-1}{3}x_3 + \binom{m-1}{4}x_4$ \\
\textbf{Subject to:}
\begin{description}
\item{Monotonicity:} $\forall_{t \in \{ 0,1, \ldots m-1 \}} : x_1 + \binom{t}{1}x_2 + \binom{t}{2}x_3 + \binom{t}{3}x_4 \geq 0$.
\item{Scaling:} $m x_1 + \binom{m}{2}x_2 + \binom{m}{3}x_3 + \binom{m}{4}x_4 = m$
\end{description}
\end{LP}

Let $f^*$ be an optimum of Linear Program~\ref{LP:min_rank4}.
Then, it must be that $f^*(m-1) < m-6$.
We consider the dual of Linear Program~\ref{LP:min_rank4}:
\begin{LP} \label{LP:max_rank4}
\ \\\textbf{Maximize} $m \cdot z$ \\
\textbf{Subject to:}
\begin{enumerate}
\item $\sum_{t=0}^{m-1} y_t + \binom{m}{1} z = \binom{m-1}{1}$
\item $\sum_{t=0}^{m-1} \binom{t}{1} y_t + \binom{m}{2} z = \binom{m-1}{2}$
\item $\sum_{t=0}^{m-1} \binom{t}{2} y_t + \binom{m}{3} z = \binom{m-1}{3}$
\item $\sum_{t=0}^{m-1} \binom{t}{3} y_t + \binom{m}{4} z = \binom{m-1}{4}$
\item $\forall_t y_t \geq 0$
\end{enumerate}
\end{LP}

Since $f^*(m-1) < n-6$ and from duality theorem, it must be that any feasible solution of Linear Program~\ref{LP:max_rank3} is of value strictly less that $m-4$.
We show a feasible solution to Linear Program~\ref{LP:max_rank3} of value $m-6$ and derive a contradiction.
Specifically, we prove that for any $n > 4$, there is a feasible solution for Linear Program~\ref{LP:max_rank4} with $z>\frac{m-6}{m}$\footnote{Of course, $z \geq \frac{m-6}{m}$ would also be enough.}. 

We consider separately even and odd values of $m$, with solutions concluded with the aid of a computer\footnote{We used Microsoft .NET together with Gurobi (\cite{gurobi}) and IBM CPLEX and also Wolfram Mathematica.} to solve Linear Program~\ref{LP:max_rank3} for various values of $m$:

\paragraph{{\boldmath $n$ is even:}}
$$
y_t = 
\begin{cases}
1 & t = 0 \\
\frac{2m-2}{m+2} & t = \frac{m}{2}-1 \\
\frac{2m-2}{m+2} & t = \frac{m}{2} \\
0 & \mbox{otherwise}
\end{cases}
$$
$$z = \frac{m-1-\sum_t{y_t}}{m} = \frac{m-4}{m+2}$$
We show that for any even $m$ this solution is feasible and that $z>\frac{m-6}{m}$.
We start with showing feasibility.
For the first equation $\sum_{t=0}^{m-1} y_t + \binom{m}{1} z = \binom{m-1}{1}$ it is trivial it is satisfied, by the value chosen for $z$.
For the second equation $\sum_{t=0}^{m-1} \binom{t}{1} y_t + \binom{m}{2} z = \binom{m-1}{2}$ we have
$$\sum_{t=0}^{m-1} t \cdot y_t = 0 \cdot 1 + \frac{m}{2}-1 \cdot \frac{2m-2}{m+2} + \frac{m}{2} \cdot \frac{2m-2}{m+2} = (m-1) \cdot \frac{2m-2}{m+2} = \frac{2(m-1)^2}{m+2}$$
and
$$\binom{m}{2} z = \frac{(1/2) \cdot m(m-1)(m-4)}{m+2} \ . $$
Therefore, 
\begin{align*}
\sum_{t=0}^{m-1} \binom{t}{1} y_t + \binom{m}{2} z =
\frac{2(m-1)^2 + (1/2) \cdot m(m-1)(m-4)}{m+2} &= \\ 
(m-1) \cdot \frac{\frac{1}{2} \cdot (m^2-4)}{m+2} &= \\
(m-1)\frac{\frac{1}{2}(m-2)(m+2)}{m+2} &= \binom{m-1}{2}
\end{align*} as desired.
For the third equation $\sum_{t=0}^{m-1} \binom{t}{2} y_t + \binom{m}{3} z = \binom{m-1}{3}$ we have
$$
\sum_{t=0}^{m-1} \binom{t}{2} y_t = \binom{\frac{m}{2}-1}{2} \cdot \frac{2m-2}{m+2} + \binom{\frac{m}{2}}{2} \cdot \frac{2m-2}{m+2} = \frac{(m-2)^2 (m-1)}{2(m+2)}
$$
and
$$
\binom{m}{3} z = \frac{m(m-1)(m-2)(m-4)}{6(m+2)} \ .
$$
Therefore,
\begin{align*}
\sum_{t=0}^{m-1} \binom{t}{2} y_t + \binom{m}{3} z = \frac{3(n-1)(n-2)^2 + m(m-1)(m-2)(m-4)}{6(m+2)} &= \\
\frac{(m-1)(m-2)(3(m-2)+m(m-4))}{6(m+2)} &= \\
\frac{(m-1)(m-2)(m+2)(m-3)}{6(m+2)} &= \binom{m-1}{3}
\end{align*} as desired.
For the fourth equation $\sum_{t=0}^{m-1} \binom{t}{3} y_t + \binom{m}{4} z = \binom{m-1}{4}$ we have
$$\sum_{t=0}^{m-1} \binom{t}{3} y_t = \binom{\frac{m}{2}-1}{3} \cdot \frac{2m-2}{m+2} + \binom{\frac{m}{2}}{3} \cdot \frac{2m-2}{m+2} = \frac{(m-1)(m-2)(m-3)(m-4)}{12(m+2)}$$
and
$$
\binom{m}{4} z = \frac{m(m-1)(m-2)(m-3)(m-4)}{24(m+2)} \ ,
$$
which immediately gives the desired.

In order to show that $z>\frac{m-6}{m}$ we show that $\sum_t y_t < 5$.
Since $z = \sum_t y_t = 1 + \frac{2(2m-2)}{m+2}$, it is enough to show that $\frac{2(2m-2)}{n+m} < 4$, which is obviously true.

\paragraph{{\boldmath $m$ is odd:}}
$$
y_t = 
\begin{cases}
1 & t = 0 \\
\frac{4(m-2)}{m+1} & t = \frac{m-1}{2} \\
0 & \mbox{otherwise}
\end{cases}
$$
$$z = \frac{(m-2)(m-3)}{m(m+1)}$$ 
We show that for any odd $m$ this solution is feasible and that $z>\frac{m-6}{m}$.
We start with showing feasibility.
For the first equation $\sum_{t=0}^{n-1} y_t + \binom{n}{1} z = \binom{n-1}{1}$ we have
$$
\sum_{t=0}^{m-1} y_t + \binom{m}{1} z = 1 + \frac{4(m-2)}{m+1} + \frac{(m-2)(m-3)}{m+1} = 1+\frac{(m-2)(m+1)}{m+1} = \binom{m-1}{1}
$$
as desired.
For the second equation $\sum_{t=0}^{n-1} \binom{t}{1} y_t + \binom{n}{2} z = \binom{n-1}{2}$ we have
$$
\sum_{t=0}^{m-1} t \cdot y_t + \binom{m}{2} z =  \frac{4(m-1)(m-2)}{2(m+1)} + \frac{(m-1)(m-2)(m-3)}{2(m+1)} = \frac{(m-1)(m-2)(m+1)}{2(m+1)} = \binom{m-1}{2}
$$
as desired.
For the third equation $\sum_{t=0}^{m-1} \binom{t}{2} y_t + \binom{m}{3} z = \binom{m-1}{3}$ we have
\begin{align*}
\sum_{t=0}^{m-1} \binom{t}{2} y_t + \binom{m}{3} z = 
\frac{3(m-1)(m-2)(m-3)}{6(m+1)} + \frac{(m-1)(m-2)^2 (m-3)}{6(m+1)} &= \\
\frac{(m-1)(m-2)(m-3)(m+1)}{6(m+1)} &= \binom{m-1}{3}
\end{align*}
as desired.
For the fourth equation $\sum_{t=0}^{m-1} \binom{t}{3} y_t + \binom{m}{4} z = \binom{m-1}{4}$ we have
\begin{align*}
\sum_{t=0}^{m-1} \binom{t}{3} y_t + \binom{m}{4} z = 
\frac{2(m-1)(m-2)(m-3)(m-5)}{24(m+1)} + \frac{(m-1)(m-2)^2(m-3)^2}{24(m+1)} &= \\
\frac{(m-1)(m-2)(m-3)(m^2-3m-4)}{24(m+1)} &= \\
\frac{(m-1)(m-2)(m-3)(m-4)(m+1)}{24(m+1)} &= \binom{m-1}{4}
\end{align*}
as desired.

We show $z>\frac{m-6}{m}$.
$$
z - \frac{m-6}{m} = \frac{(m-2)(m-3)}{m(m+1)} - \frac{m-6}{m} = \frac{(m-2)(m-3) - (m-6)(m+1)}{m(m+1)} = \frac{12}{m(m+1)} > 0
$$
where the second and third equalities follow from $m>0$.

This concludes the proof of Theorem~\ref{t:upperBoundRank4}.
\end{proof}

\begin{theorem} \label{t:tightnessRank4}
There exists a monotone symmetric set function of rank~4 that has no positive lower envelope of rank~5.
\end{theorem}
\begin{proof} 
Let $f: \{0, \ldots, m\} \to \nonnegR$ be a set function with the following hypergraph representation.
Hyperedges of rank~1 have value~0;
Hyperedges of rank~2 have value~10;
Hyperedges of rank~3 have value~-8;
Hyperedges of rank~4 have value~3.
That is, $f(x) = 10 \binom{x}{2} - 8\binom{x}{3} + 3\binom{x}{4}$.
It is clear this function is symmetric and that it has rank~4.
This function is monotone, since from simple calculation the marginal value is non-negative for any integer (actually it is positive for any $m>7$).
We show inexistence of positive lower envelope of rank~5.
By Corollary~\ref{c:SymmetricPleClassificationByN-1} it is sufficient to show that $f(m-1) < \frac{m-5}{m}f(m)$.
We show it is true for $m=12$:
$$f(m)=f(12)=385$$
$$f(m-1)=f(11)=220$$
$$\frac{m-5}{m}f(m) = \frac{7}{12} \cdot 385 = \frac{2695}{12} \approx 224.583 > 220 = f(m-1) \ , $$
as desired.
This concludes the proof of Theorem~\ref{t:tightnessRank4}.
\end{proof}

\section{Limitations of $\MOPH$}\label{sec:limitations}
Despite the advantages of the $\MOPH$ hierarchy, it has a few limitations.
First, even a function in the lowest level of the hierarchy may need an exponential number of hypergraphs in its support.

\begin{prop}
Every $\MOPH$ representation (regardless of rank) of the submodular function $f(S) = \min[|S|,n/2]$ (which is in $\MOPH$-1) requires exponentially many hypergraphs.
\end{prop}

\begin{proof}
Consider two different sets $S$ and $T$ with $|S| = |T| = m/2$. Let $h$ be an arbitrary supermodular function (hence any positive hypergraph qualifies here) satisfying $h(Q) \le f(Q)$ for every set $Q$. We claim that either $h(S) \not= f(S)$ or $h(T) \not= f(T)$. This proves that at least ${m \choose m/2}$ hypergraphs are needed.

To prove the claim observe that $f(S \cap T) < m/2$ and $f(S \cup T) = m/2$. Suppose for the sake of contradiction that $h(S) = h(T) = m/2$. Then by supermodularity  $h(S \cap T) + h(S \cup T) \ge m$, implying that either $h(S \cap T) > f(S \cap T)$ or $h(S \cup T) > f(S \cup T)$, a contradiction.
\end{proof}

Second, there are complement-free functions that can only be represented in level $m/2$ of the $\MOPH$ hierarchy.

\begin{prop}
There exists a subadditive function that cannot be represented by an $\MOPH$-$k$ function for any $k<m/2$, when $m$ is even.
\end{prop}

\begin{proof}
For a ground set of even size $m$, consider the function $f$ that gives value $1$ for every subset except for the set of all $m$ items, for which the value is $2$.
This function is clearly subadditive.
Additionally, for any rank $k < m/2$, it follows by the ``moreover" part of Corollary~\ref{c:SymmetricPleClassificationByN-1} and by straightforward calculations that $f$ does not have a positive lower envelope of rank $k$.
Finally, the canonical positive lower envelope of rank $m/2$ of $f$ (see Lemma~\ref{l:canonicalPLE_Symmetric}) is a legal positive lower envelope of it.
\end{proof}

\section{Integrality gap of LP~\ref{LP:configuration}} \label{app:integralityGap}
\begin{prop} \label{t:integralityGap}
Let $k \in \mathbb{N}$ be such than $k-1$ is a power of prime.
There exists an instance of the welfare maximization problem with $\PH$-$k$ valuations
and integrality gap $k - 1 + \frac{1}{k}$ for Linear Program~(\ref{LP:configuration}).
Note that such an instance is in particular $\MOPH$-$k$.
\end{prop}

\begin{proof}
Let $FPP_{k-1}$ be the finite projective plane of order $k-1$
(it is known to exist, since $k-1$ is a power of prime). 
We set the following hypergraph $H=(V,E)$.
For each point in $FPP_{k-1}$, we have a vertex in $V$, and for each line, we have a hyperedge in $E$, containing the vertices representing the points that are on this line.
The following follows from the definitions of finite projective planes:
\begin{itemize}
\item $|V| = |E| = (k-1)^2+(k-1)+1 = (k-1)k + 1$.
\item Any two hyperedges in $E$ have a vertex in common.
\item Any hyperedge in $E$ contains exactly $k$ vertices (i.e. the hyperedges in $E$ are all of rank $k$).
\item Any vertex in $V$ is contained in exactly $k$ hyperedges.
\end{itemize}
Our instance of the welfare maximization problem has $(k-1)^2+k$ agents.
Each agent $i$ has one distinct preferred hyperedge $e_i \in E$.
The valuation function $v_i$ of agent $i$ has value~1 for any subset containing all the items represented by vertices in $e_i$ and~0 otherwise.
It is trivial that $v_i$ is in $\PH$-$k$. 
Furthermore, the agents are single minded.
An optimal {\em integral} solution of this instance is an allocation of all the items represented by vertices in $e_i$ to agent $i$, for some arbitrary $i$.
This solution has value of~1.
However, there exists a better {\em fractional} solution.
Any agent $i$ gets a fraction of $\frac{1}{k}$ of all the items represented by vertices in $e_i$.
It is easy to verify that this is a feasible fractional solution with value $((k-1)k+1)/k = k - 1 + \frac{1}{k}$, as desired.
\end{proof}

\section{Complement-Free Valuations}\label{s:typesOfSetFunctions}
For completeness, we present the hierarchy of complement-free valuations (see for example~\cite{Lehmann2001, Feige2006}).
Let $M$ be a ground set and let $f: 2^M \to \nonnegR$ be a set function.

\begin{defn} [Additive function]
We say that $f$ is {\em additive} or {\em linear} if for every subset $S' \subseteq M$ of items, we have $f(S') = \sum\limits_{j\in S'} f(\{j\})$.
\end{defn}

\begin{defn} [Submodular function]
We say that $f$ is {\em submodular} if for every $S'' \subseteq S' \subseteq M$
and $x \in M \setminus S'$,
$f(x \mid S') \leq f(x \mid S'')$.
\end{defn}

\begin{defn} [XOS]
We say that $f$ is in {\em $\XOS$} if for some $l \in \mathbb{N}$ there exist additive set functions $f_1, \cdots f_l$ such that for every $S' \subseteq M$, we have $f(S') = \max_{1 \leq i \leq l} f_i(S')$.
\end{defn}

\begin{defn} [Fractionally subadditive function]
We say that $f$ is {\em fractionally subadditive} if for every subset $S' \subseteq M$,
subsets $T_i \subseteq S'$ and
every coefficients $0 < \alpha_i \leq 1$
such that for every $x \in S'$, $\sum_{i: x \in T_i} \alpha_i \geq 1$,
it holds that $f(S') \leq \sum_i \alpha_i f(T_i)$.
\end{defn}

\begin{defn} [Subadditive function]
We say that $f$ is {\em subadditive} or {\em complement free} if for every $S_1, S_2 \subseteq M$,
$f(S_1 \cup S_2) \leq f(S_1) + f(S_2)$.
\end{defn}

Note that Additive $\subseteq $ Submodular $\subseteq \XOS = $ Fractionally subadditive $ \subseteq $ Subadditive.

\section{Types of queries} \label{app:queries}
We recall the definitions of basic types of queries for set functions.
Let $M$ be a ground set and let $f: 2^M \to \nonnegR$ be a set function.
\begin{defn} [Value queries]
{\em Value queries} are the following:
\\\textsf{Input:} A subset $S' \subseteq M$.
\\\textsf{Output:} $f(S')$.
\end{defn}
\begin{defn} [Demand queries]
{\em Demand queries} are the following:
\\\textsf{Input:} A cost function $c: M \to \nonnegR$.
\\\textsf{Output:} A subset $S' \subseteq M$ maximizing $f(S') - \sum_{j \in S'} c(j)$.
\end{defn}
Note that demand queries are strictly stronger than value queries (see~\cite{Nisanagt}).

\end{document}